\documentclass{sig-alternate}

\usepackage[latin1]{inputenc}
\usepackage{makeidx}  % allows for indexgeneration
\usepackage{amsmath,amssymb,amsfonts,latexsym,epsfig} 
\usepackage{psfrag, xspace}
\usepackage{hyperref}
\usepackage{framed}
\usepackage{bm}
\usepackage{enumitem}
\usepackage{algorithmic} 
\usepackage[linesnumbered,vlined,boxruled,boxed,algo2e]{algorithm2e}

\usepackage{notations}

\usepackage{hyperref}
\newtheorem{theorem}{Theorem}
\newtheorem{lemma}[theorem]{Lemma}

\newtheorem{proposition}[theorem]{Proposition}
\newtheorem{corollary}[theorem]{Corollary}
\newtheorem{remark}[theorem]{Remark}
\newtheorem{problem}{Problem}

\newtheorem{fact}[theorem]{Fact}

\newcommand{\x}{\ensuremath{\mathbf{x}}\xspace}

\newcommand{\lD}{\ensuremath{\lg\prod_i\Delta_i}\xspace}

\newcommand{\wid}{\mathtt{wid}\xspace}

\def\wt#1{\widetilde #1}

\def\normi#1{\norm{#1}_{\infty}}
\def\normo#1{\norm{#1}_{1}}

\def\Normo#1{\Norm{#1}_{1}}

\def\diag{\mathsf{diag}}

\def\smallddots{\mathinner{\raise7pt\hbox{.}\raise4pt\hbox{.}\raise1pt\hbox{.}}}
\def\smallsdots{\mathinner{\raise1pt\hbox{.}\raise4pt\hbox{.}\raise7pt\hbox{.}}}

\usepackage[paper=letterpaper,width=7in,height=9.25in,centering]{geometry}

\begin{document}

% \conferenceinfo{ISSAC'13,} {June 26--29, 2013, Boston, Massachusetts, USA.}
% \CopyrightYear{2013}
% \crdata{978-1-4503-2059-7/13/06}
% \clubpenalty=10000
% \widowpenalty = 10000

\title{Nearly Optimal Computations with Structured Matrices}

\numberofauthors{2}
\author{
  \alignauthor Victor Y.~Pan \\
  \affaddr{Depts. of Mathematics and Computer Science \\
    Lehman College and Graduate Center \\
    of the City University of New York \\
    Bronx, NY 10468 USA} 
%
%$^{[2]}$ Ph.D. Programs in Mathematics  and Computer Science \\
%The Graduate Center of the City University of New York \\
%New York, NY 10036 USA \\
%$^{[a]}$ victor.pan@lehman.cuny.edu \\
%http://comet.lehman.cuny.edu/vpan/  
%
%
  \email{victor.pan@lehman.cuny.edu}
  \url{http://comet.lehman.cuny.edu/vpan/} 
  \alignauthor Elias P.~Tsigaridas \\
  \affaddr{ 
    PolSys~Project\\
    INRIA, Paris-Rocquencourt~Center\\
    UPMC, Univ Paris 06, LIP6 \\
    CNRS, UMR 7606, LIP6 \\
    Paris, France
    }
  % \affaddr{Computer Science Department \\Aarhus University, Denmark}\\
  \email{elias.tsigaridas@inria.fr}
}

\maketitle

\begin{abstract}
  We estimate the Boolean complexity of multiplication of structured
  matrices by a vector and the solution of nonsingular linear systems
  of equations with these matrices. We study four basic most popular
  classes, that is, Toeplitz, Hankel, Cauchy and Van\-der\-monde
  matrices, for which the cited computational problems
%and their solution algorithms 
are equivalent to the task of polynomial multiplication and division
and polynomial and rational multipoint evaluation and
interpolation. The Boolean cost estimates for the latter problems have
been obtained by Kirrinnis in \cite{kirrinnis-joc-1998}, except for
rational interpolation, which we supply now. All known Boolean cost
estimates for these problems rely on using Kronecker product. This
implies the $d$-fold precision increase for the $d$-th degree output,
but we avoid such an increase by relying on distinct techniques based
on employing FFT. Furthermore we simplify the analysis and make it
more transparent by combining the representation of our tasks and
algorithms in terms of both structured matrices and polynomials and
rational functions.  This also enables further extensions of our
estimates to cover Trummer's important problem and computations with
the popular classes of structured matrices that generalize the four
cited basic matrix classes.
  \end{abstract} 

%  \vspace{10pt}

% \noindent
% {\bf Categories and Subject Descriptors:}\\
% F.2 [Theory of Computation]: Analysis of Algorithms and Problem Complexity;
% I.1 [Computing Methodology]: Symbolic and algebraic manipulation: Algorithms

% \vspace{0.9mm}
% \noindent
% {\bf General Terms} \\
% Algorithms, Theory

% \vspace{0.9mm}
% \noindent
% {\bf Keywords}
% Boolean comlexity;
% approximate computations;
% polynomials;
% fast polynomial division;
% rational functions;
% multipoint evaluation; interpolation; precision of computing

\section{Introduction}
{
\scriptsize
\begin{table*}[ht]
\caption{Four classes of structured matrices}
\label{t1}.
\begin{center}
\renewcommand{\arraystretch}{1.2}
\begin{tabular}{c|c}
Toeplitz matrices $T=\left(t_{i-j}\right)_{i,j=0}^{n-1}$&Hankel matrices $H=\left(h_{i+j}\right)_{i,j=0}^{n-1}$ \\
$\begin{pmatrix}t_0&t_{-1}&\cdots&t_{1-n}\\ t_1&t_0&\smallddots&\vdots\\ \vdots&\smallddots&\smallddots&t_{-1}\\ t_{n-1}&\cdots&t_1&t_0\end{pmatrix}$&$\begin{pmatrix}h_0&h_1&\cdots&h_{n-1}\\ h_1&h_2&\smallsdots&h_n\\ \vdots&\smallsdots&\smallsdots&\vdots\\ h_{n-1}&h_n&\cdots&h_{2n-2}\end{pmatrix}$ 
\\ \hline 

%\medskip

Van\-der\-monde matrices 
$V=V_{\bf s}=\left(s_i^{j}\right)_{i,j=0}^{n-1}$&Cauchy matrices $C=C_{\bf s,t}=\left(\frac{1}{s_i-t_j}\right)_{i,j=0}^{n-1}$ \\
$\begin{pmatrix}1&s_1&\cdots&s_1^{n-1}\\ 1&s_2&\cdots&s_2^{n-1}\\ \vdots&\vdots&&\vdots\\ 1&s_{n}&\cdots&s_{n}^{n-1}\end{pmatrix}$&$\begin{pmatrix}\frac{1}{s_1-t_1}&\cdots&\frac{1}{s_1-t_{n}}\\ \frac{1}{s_2-t_1}&\cdots&\frac{1}{s_2-t_{n}}\\ \vdots&&\vdots\\ \frac{1}{s_{n}-t_1}&\cdots&\frac{1}{s_{n}-t_{n}}\end{pmatrix}$
\end{tabular}
\end{center}
\end{table*}
}

\vspace{10pt}

Table \ref{t1} displays four classes of most popular structured matrices,
which
are omnipresent in modern computations for Sciences, Engineering, and Signal
and Image Processing.  
These basic classes have been naturally extended to the four larger classes of matrices, 
$\mathcal T$, $\mathcal H$, $\mathcal V$, and $\mathcal C$,
that have structures of Toep\-litz, Hankel,
Van\-der\-monde and Cauchy types,
 respectively. They include many other important classes of structured matrices such as the products and inverses of the matrices of these four basic classes, as well as the companion, Sylvester, subresultant, Loewner, and Pick matrices. All these 
matrices 
%of these classes
%The computations with these matrices are closely linked to
%computations with polynomials and rational functions.
can be readily expressed via
their displacements of small ranks \cite[Chapter 4]{Pan01}, which implies
their further attractive properties:

\begin{itemize}[leftmargin=5pt,rightmargin=3pt]
\item
Compressed representation of matrices as well as their products and inverses through 
a small number of parameters.
\item
Multiplication by a vector in nearly linear arithmetic time.
\item
Solution of nonsingular
linear systems of equations with these matrices
in quadratic or nearly linear arithmetic time.
\vspace{-2mm}
\end{itemize}
These properties enable efficient
computations,   
closely linked and frequently equivalent to fundamental computations with polynomials and rational  polynomial functions, in particular to the multiplication, division,
multipoint evaluation and interpolation \cite{pt-arxiv-refine-j-14}.
Low arithmetic cost is surely attractive,
 but substantial growth of the computational precision quite  frequently affects the known algorithms
having low arithmetic cost 
(see, e.g., \cite{B85}).
%, \cite{BF90}). 
So the estimation of the complexity under the Boolean model  is  more  informative, although technically  more demanding.

To the best of our knowledge,
%vy
the first Boolean complexity bounds for multipoint evaluation
are due to Ritzmann~\cite{Ritmann-tcs-1986}.
We also wish to cite the papers \cite{vdH:fastcomp}
and \cite{ks-arxiv-mpeval-2013}, 
although their results have been superceded 
in the advanced work of 1998 by Kirrinnis, 
 \cite{kirrinnis-joc-1998},
apparently still not sufficiently well known.
Namely in
the process of studying 
approximate partial fraction decomposition
he
has
estimated the Boolean complexity of
the multipoint evaluation, interpolation, and the summation of
rational 
%vy polynomial 
functions.  He required the input polynomials to
be normalized, but actually this was
not restrictive at all.  
We 
%vy revisit and 
generalize his estimates.
For simplicity we assume the evaluation at the 
 points of small magnitude, but our 
%y bounds 
estimates
can be rather easily 
extended to the case of general input. Kirrinnis' study as well as all previous estimates of 
the Boolean complexity of these  computational problems
rely on multiplying polynomials as integers, by using Kronecker's product, 
aka binary segmentation, as proposed in \cite{FP74}.
This implies the
$d$-fold increase of the computational precision for the $d$-th degree output. 
The results that we present rely on FFT algorithms for multiplying 
univariate polynomials and avoid this precision growth.

We represent our FFT-based estimates and algorithms in terms of
operations with both structured matrices and polynomial and rational
functions. In both representations the computational tasks and the
solution algorithms are equivalent, and so the results of
\cite{kirrinnis-joc-1998} for partial fraction decomposition can be
extended to most although not all of these tasks.  By using both
representations, however, we make our analysis more
transparent. Furthermore in Section~\ref{sec:Cauchy-solve} we extend
Kirrinnis' results to the solution of a Cauchy linear system of
equations (which unlike \cite{kirrinnis-joc-1998} covers rational
interpolation) and in Section~\ref{sec:Trummer} to the solution of
Trummer's celebrated problem \cite{GGS87}, \cite{GR87}, \cite{CGR88},
having important applications to mechanics (e.g., to particle
simulation) and representing the secular equation, which is the basis
for the MPSolve, the most efficient package of subroutines for
polynomial root-finding \cite{BRa}.

Our estimates cover multiplication of the matrices of the four basic
classes of Table \ref{t1} by a vector and solving Van\-der\-monde and
Cauchy linear systems of equations.  (These tasks are equivalent to
the listed tasks of the multiplication, division, multipoint
evaluation and interpolation of polynomials and rational functions.)
Expressing the solution of these problems in terms of matrices has a
major advantage: it can be extended to matrices from the four larger
matrix classes $\mathcal T$, $\mathcal H$, $\mathcal V$, and $\mathcal
C$. Actually the algorithms for multiplication by vector can be
extended quite readily, as we explain in
Section~\ref{sec:Cauchy-solve}. There we also briefly discuss the
solution of linear systems of equations with the matrices of the cited
classes, which can be a natural subject of our further study.

\paragraph*{Notation}
In what follows \OB, resp. \OO, means bit, resp. arithmetic,
complexity and  \sOB, resp. \sO,  means that we are
ignoring logarithmic factors. ``Ops" stands for ``arithmetic operations".
%???
For a polynomial $A = \sum_{i=0}^{d}{a_i \, x^i} \in$ $\ZZ[x]$,
$\dg{A} = d$ denotes its degree and $\bitsize{A} =\tau$ the maximum bitsize of
its coefficients, including a bit for the sign. 
For $a \in \QQ$, $\bitsize{ a} \ge 1$
is the maximum bitsize of the numerator and the denominator.
%%???
$\mu(\lambda)$ denotes the bit complexity of multiplying two integers of size
$\lambda$; 
we have  $\mu(\lambda) = \sOB( \lambda )$.
$2^{\Gamma}$ is an upper bound on the magnitude of the roots of
$A$. We write $\Delta_{\alpha}(A)$ or just  $\Delta_{\alpha}$ to denote 
the minimum distance between a root $\alpha$ of a polynomial $A$ and 
any other root. We call this quantity {\em local separation bound}.
We also write $\Delta_i$ instead of $\Delta_{\alpha_i}$.
$\Delta(A)=\min_{\alpha}{\Delta_{\alpha}(A)}$ or just $\Delta$ denotes
the {\em separation bound}, that is the minimum distance between all
the roots of $A$.
The Mahler bound (or measure) of $A$ is
$\Mahler{A} = a_d \prod_{\abs{\alpha} \geq 1}{\abs{\alpha}}$,
where $\alpha$ runs through the complex roots of $A$, 
e.g. \cite{Mign91,Yap2000}.
If 
$A\in \ZZ[x]$ and $\bitsize{A}=\tau$,
then 
$\Mahler{A} \leq \norm{A}_2 \leq \sqrt{d+1} \norm{A}_{\infty} 
= 2^{\tau} \sqrt{d+1} $.
If we evaluate a function $F$ (e.g. $F=A$) at a number $c$ using
interval arithmetic, then we denote the resulting interval by
$[F(c)]$, provided that we fix the evaluation algorithm and the
precision of computing.
We write $D(c, r) = \{ x \,:\, \abs{x-c} \leq r \}$.
$\wt{f} \in \CC[x]$ denotes a $\lambda$-approximation
to a polynomial $f \in \CC[x]$,
such that  $\normi{f - \wt{f}} \leq 2^{-\lambda}$.
In particular $\wt{a} \in \CC$ denotes a $\lambda$-approximation
to a constant $a \in \CC$ such that
$\abs{a - \wt{a}} \leq 2^{-\lambda}$.
$\lg$ stands for $\log$.

\section{Preliminaries}

\subsection{Univariate Separation Bounds}

The following proposition provides upper and aggregate bounds for the
roots of a univariate polynomial.  
There are various version of these bounds. We use the one presented 
in \cite{st-issac-2011},
to which we also refer the reader for further details and a discussion 
of the literature. For
multivariate separation bounds we refer the reader
to \cite{emt-issac-2010}.

\begin{proposition}
  \label{prop:sep-bounds}
  Let $f = \sum_{i=0}^{d}{a_i x^{i}} \in \CC[x]$ be a square-free univariate polynomial of 
  a degree $d$
  such that $a_d a_0 \not= 0$.
  Let $\Omega$ be any set of $k$ pairs of indices $(i, j)$ such that
  $1 \leq i < j \leq d$,
  let the complex roots of $A$ be 
  $0 < | \gamma_1 | \leq |\gamma_2| \leq \dots \leq |\gamma_d|$,
  and let $\disc(f)$ be the discriminant of $f$.
  Then
  {\footnotesize
  \begin{equation}
    \frac{\abs{a_0}}{\norm{f}_2} \leq |\gamma_i| \leq  \frac{\norm{f}_{2}}{\abs{a_d}}
    \label{eq:u-upper} \enspace, 
  \end{equation}
  \begin{equation}
    \begin{aligned}
      \prod_{(i, j) \in \Omega}{|\gamma_i - \gamma_j|} & \geq 
      2^{k - d -\frac{d(d-1)}{2}} \,\abs{a_0}^{k} \, \mathcal{M}(f)^{1-d-k} \,\sqrt{| \disc(f) |} \\
      & \geq 2^{k - d -\frac{d(d-1)}{2}} \,\abs{a_0}^{k} \, \norm{f}_2^{1-d-k} \,\sqrt{| \disc(f) |} .
      \\
    \end{aligned}
    \label{eq:u-dmm}  \enspace                                                                                                                  
  \end{equation}
  }
  If $f \in \ZZ[x]$ and the maximum coefficient bitsize is $\tau$ then
  \begin{equation}
    2^{-\tau - 1} \leq \abs{\gamma} \leq 2^{\tau + 1} \enspace,
      \label{eq:upper-z}
  \end{equation}
  \begin{equation}
    - \lg \prod_{(i, j) \in \Omega}{|\gamma_i - \gamma_j|}  \leq  
      3d^2 + 3d \tau + 4d\lg{d}.
      \label{eq:dmm-z}
      \enspace
  \end{equation}

\end{proposition}

The following lemma from \cite{st-arxiv-rslog-14} provides a lower
bound on the evaluation of a polynomial that depends on the closest
root and on aggregate separation bounds.
\begin{lemma}
  \label{lem:poly-eval-lower-bound} 
  Suppose  $L \in \CC$,   
  $f$ is a square-free polynomial, and 
  its root $\gamma_1$ is closest to $L$. 
  Then
  \begin{displaymath}
    \abs{f(L)} \geq \abs{a_d}^7 \,\abs{L - \gamma_1}^{6} \, \Mahler{f}^{-6} 2^{\lD-6}
    \enspace.
  \end{displaymath}
\end{lemma}

\subsection{Complex Interval arithmetic}

We also need the following bounds for the width of complex intervals when we perform
computations with interval arithmetic.
We will use them to bound the error
when we perform basic computation with complex (floating point) numbers.
We refer the reader to \cite{RoLa-cacm-71} for further details.
\begin{proposition}[Complex intervals]
  \label{prop:complex-intvl}
  Given complex intervals $I$ and $J$,
  where $\abs{I}$, resp. $\abs{J}$, denotes
  the modulus of any complex number in the complex interval $I$, resp. $J$.
  If $2^{-\nu} \leq \abs{I} \leq 2^{\tau}$ and $\abs{J} \leq 2^{\sigma}$,
  then 
  $\wid(I + J) \leq 2\,\wid(I) + 2\,\wid(J)$,
  $\wid(I \, J) \leq 2^{\tau+1} \, \wid(J) + 2^{\sigma+1}\, \wid(I)$,
  and 
  $\wid(1/I) \leq 2^{4\nu +2\tau + 3} \,\wid(I)$.
\end{proposition}

\subsection{Approximate multiplication of two  polynomials}
\label{sec:a-poly-mul}

We need the following two 
lemmas from \cite{pt-arxiv-refine-j-14} on the evaluation of a polynomial at the
powers of a root of unity and on polynomial multiplication. 
A result  similar to the first lemma appeared
in \cite[Section~3]{Schon-fast-mult:82}
where Bluestein's technique from
\cite{b-ieee-970} is applied
(see also \cite[Chapter~4.3.3, Exercise~16]{Knuth-art-2-2nd}).
We use that lemma to provide a bound on the Boolean complexity
of multiplying two univariate polynomials
when their coefficients are known up to a
fixed precision.
An algorithm for this problem appeared in \cite[Theorem~2.2]{Schon-fast-mult:82}
based on employing Kronecker's product, 
but instead we rely on FFT and the estimates of Corollary 4.1 from \cite[Chapter 3]{BiPa}.

\begin{lemma}
  \label{lem:a-fft}
  Suppose $A \in \CC[x]$ of a degree at most $d$
  such that   $\norm{A}_{\infty} \leq 2^{\tau}$. Let $K=2^k\ge d$
for a positive integer $k$.
  Assume that we know the coefficients of $A$ up to the precision
  $-\ell - \tau - \lg{K} - 3$;
  that is the input is assumed to be a polynomial $\wt{A}$ such that 
  $\norm{A -\wt{A}}_{\infty} \leq 2^{-\ell - \tau - \lg{K} - 3}\ge 10$.
  Let $\omega=\exp(\frac{2\pi}{K}\sqrt{-1})$ denote a $K$-th root of unity. 
  Then we can evaluate the polynomial $A$ at $1, \omega, \dots, \omega^{K-1}$
  in $\sOB( K \lg{K} \, \mu( \ell + \tau + \lg{K}))$
  such that 
  $ \max_{0\leq i \leq K-1} \abs{A(\omega^{i}) - \widetilde{A(\omega^{i})}} \leq 2^{-\ell} $.
  Moreover, $\abs{A(\omega^{i})} \leq  K \, \norm{A}_{\infty} \leq 2^{\tau + \lg{K}}$,
  for all $0 \leq i \leq K-1$.
\end{lemma}

\begin{lemma}
  \label{lem:a-poly-mult}
  Let $A, B \in \CC[x]$ of degree at most $d$,
  such that   $\norm{A}_{\infty} \leq 2^{\tau_1}$
 and 
  $\norm{B}_{\infty} \leq 2^{\tau_2}$.
  Let $C$ denote the product $A B$ and let
  $K =2^k\ge 2 d + 1$ for a positive integer $k$.
 Write $\lambda= \ell +  2\tau_1 + 2\tau_2 + 5.1\lg{K} + 4$.
  Assume that we know the coefficients of $A$ and $B$ up to 
the precision $\lambda$,
  that is that the input includes two polynomials $\wt{A}$ and $\wt{B}$ such that 
  $\norm{A -\widetilde A}_{\infty} \leq 2^{-\lambda}$
  and $\norm{B -\widetilde B}_{\infty} \leq 2^{-\lambda}$.
  Then we can compute in 
  $\OB( d \lg d \,  \mu( \ell + \tau_1 + \tau_2 + \lg d))$
  a polynomial  $\widetilde C$ such that
  $ \norm{C - \wt{C}}_{\infty}\leq 2^{-\ell}$. 
  Moreover, $ \norm{C}_{\infty} \leq 2^{\tau_1 + \tau_2 + 2\lg{K}}$ for all $i$. 
\end{lemma}

\begin{remark}
  \label{rem:a-poly-mult}
  % For simplicity we also use 
  % $\lambda = \ell + 2\tau_1 + 2\tau_2 + 
  % %% V
  % 5.1\lg{d} + 
  % %% V
  % 4$ 
  %% NO it should be as follows:
  In the sequel, for simplicity 
we occasionally replace the value
  $\lambda = \ell + 2\tau_1 + 2\tau_2 + 5.1\lg(2d + 1) + 4$
by its simple upper bound
% \leq
%  \ell + 2\tau_1 + 2\tau_2 + 5.1\lg{d} + 14.2 \leq
  $\ell + 2\tau_1 + 2\tau_2 + 6\lg{d} + 15$ .
\end{remark}

\section{Approximate FFT-based polynomial  division}
\label{sec:a-poly-div}

In this section we present an efficient algorithm and its complexity analysis 
for dividing univariate polynomials
approximately.  
This result is the main ingredient
of the fast algorithms for multipoint evaluation and interpolation.
The evaluation is involved into our record fast real root-refinement,
 but all these results  are also
interesting on 
their own right because, unlike the previous
papers such as \cite{Schon-fast-mult:82}, \cite{Sch82} and \cite{kirrinnis-joc-1998}, we keep 
the Boolean cost bounds of these computations 
at the record level
by employing FFT rather than the Kronecker product
and thus decreasing the precision of computing dramatically.

%!!!! {\bf ELIAS, THE FIRST TWO PAGES OF THIS SECTION REPEAT OUR PAPER FOR JSC. IS THIS OK?

%ET: I think that we should say this explicity.
%I think that we need the details to make the result readable.

Assume two polynomials $s(x) = \sum_{i=0}^{m}s_i x^{i}$ and
$t(x)=\sum_{i=0}^{n}t_i x^i$ such that $s_mt_n \not= 0$, $m\ge n$, and
seek the quotient $q(x) = \sum_{i=0}^{m-n}{q_i x^{i}}$ and the
remainder $r(x) = \sum_{i=0}^{n-1}{r_i x^{i}}$ of their division such
that $s(x) = t(x) \, q(x) + r(x)$ and $\deg (r)< \deg (t)$.  Further
assume that $t_{n}=1$. This is no loss of generality because we can
divide the polynomial $t$ by its nonzero leading coefficient.  We
narrow our task to computing the quotient $q(x)$ because as soon as
the quotient is available, we can compute the remainder $r(x) =
s(x)-t(x) \, q(x)$ at the dominated cost by multiplying $t(x)$ by
$q(x)$ and subtracting the result from $s(x)$.

The complexity analysis that we present relies on root bounds of
$t(x)$, contrary to \cite{pt-arxiv-refine-j-14} where it relies on
bounds on the infinity norm of $t(x)$.  To keep the presentation
self-contained we copy from \cite{pt-arxiv-refine-j-14} 
the matrix representation of the algorithm, which occupies 
the next two pages, up to to Lemma~\ref{lem:quo-rem-bd}.

We begin with an algorithm for the exact evaluation of the quotient. 
Represent division with a remainder by the
vector equation
{\scriptsize
\begin{displaymath}
  \left[
    \begin{array}{llllllllllll}
      1      &   \\
      t_{n-1}   & 1 \\
      \vdots   &  \vdots \\
      t_1 & \\
      t_0 & t_1 &  \cdots  &1\\  %\hline
      & t_0 & t_1 &  \vdots \\
      %& \\
      & & & t_1\\
      & & & t_0
    \end{array}
  \right ]
  \left [
    \begin{array}{l}
      q_{m-n} \\
      q_{m-n-1} \\
      \vdots  \\
      q_1 \\
      q_0
    \end{array}
  \right]
  +
  \left [
    \begin{array}{l}
      \\
      \\
      \\
      \\ %\hline
      r_{n-1}\\
      r_{n-2}\\
      \vdots  \\
      r_0\\
    \end{array}
  \right]
  =
  \left [
    \begin{array}{l}
      s_{m} \\
      s_{m-1} \\
        \vdots  \\
      s_n \\ %\hline
      s_{n-1} \\
     \vdots  \\
      \\
      s_0\\
    \end{array}
  \right]
  \enspace .
\end{displaymath}
}
The first $m-n+1$ equations form 
 the following vector equation,
{\scriptsize 
\begin{equation}
  \label{eq:T-lin-eq}
  \left[
    \begin{array}{llllllllllll}
      1      &   \\
      t_{n-1}   & 1 \\
      \vdots   &  \vdots \\
      t_1 & \\
      t_0 & t_1 &  \cdots & &1\\  %\hline
    \end{array}
  \right ]
  \left [
    \begin{array}{l}
      q_{m-n} \\
      q_{m-n-1} \\
      \vdots \\
      q_1 \\
      q_0\\ %\hline
    \end{array}
  \right]
  =
  \left [
    \begin{array}{l}
      s_{m} \\
      s_{m-1} \\
      \vdots\\
      s_{n+1}\\
      s_n \\ %\hline
      \vdots \\
      s_{m-n-1}
    \end{array}
  \right]
  \Leftrightarrow 
  T \, {\bf q} = {\bf s} ,
  \enspace 
\end{equation}
} %
%Rewrite this system as the vector equation ${\bf s}=T{\bf q}$ 
where
 ${\bf q}=(q_i)_{i=0}^{m-n}$,
 ${\bf s}=(s_i)_{i=m-n+1}^m$,
and $T$ is the nonsingular lower triangular Toeplitz matrix,
defined by its first column vector ${\bf t}=(t_i)_{i=0}^{n})$, $t_n=1$.
Write $T=Z({\bf t})$ and $Z=Z({\bf e}_2)$ where ${\bf e}_2=(0,1,0\dots,0)^T$
is the second coordinate vector, and express the matrix $T$ as a polynomial in 
a generator matrix $Z=Z_{n+1}$ of size $(n+1)\times (n+1)$ as follows,
{\scriptsize
$$Z=\begin{pmatrix}
        0   &       &   \dots    &   & 0\\
        1   & \ddots    &       &   & \\
        \vdots     & \ddots    & \ddots    &   & \vdots    \\
            &       & \ddots    & 0 &  \\
        0    &       &  \dots      & 1 & 0 
    \end{pmatrix},~~ T=Z({\bf t})=t(Z)=\sum_{i=0}^nt_iZ^i,~~Z^{n+1}=O.
$$
}
The matrix $T$ is nonsingular because $t_n\neq 0$, and the latter 
equations imply that the inverse matrix $T^{-1}=t(Z)^{-1}\mod Z^{n+1}$ is again a polynomial in 
 $Z$, that is again a lower triangular Toeplitz matrix defined by its first column.
We compute this column by applying a divide and conquer algorithm. 
Assume that 
%!!!
$n+1=\gamma=2^k$ is a power of two, for a positive integer $k$.
If this is not the case, embed the matrix $T$ into  a lower triangular Toeplitz
%!!!
$\gamma\times \gamma$ matrix $
%!!!
\bar t(Z_{\gamma})$ for
$\gamma=2^k$ and $k=\lceil \lg (n+1)\rceil$ 
 with the leading (that is northwestern) block 
%!!!
$T=t(Z_{\gamma})$,
such that 
%!!!
$t(Z_{\gamma})=\bar t(Z_{\gamma}) \mod Z_{\gamma}^{n+1}$,
compute the inverse matrix 
and output its leading  $(n+1)\times (n+1)$
block $T^{-1}$.

Now represent $T$ as the $2\times 2$
block matrix,  
 $ T = 
  \left[
    \begin{array}[c]{ccc}
      T_0 & 
      O\\ 
\noalign{\vspace{10pt}}
      T_1 & T_0
    \end{array}
  \right ]$ 
where $T_0$ and $T_1$ are
%!!!
 $\frac{\gamma}{2} \times \frac{\gamma}{2}$ Toeplitz submatrices of 
the Toeplitz matrix $T$,
$T_0$ is invertible,
and observe that
%Note that general $q\times q$ Toeplitz matrix $T_{\gamma}=(t_{i-j})_{i,j=1}^{\gamma}$
%is expressed via the vector ${\bf t}_+=(t_i)_{i=1-\gamma}^{\gamma-1}$ made up of the entries 
%of the first row and column
%where we abuse notation and  we also use $\gamma$ and $s$ to represent the corresponding vectors.
%To compute $\gamma$ it suffices to invert the matrix $T$, and then $q = T^{-1} \, s$.
\begin{equation}
  \label{eq:T-inv}
  T^{-1} = 
  \left[
    \begin{array}[c]{ccc}
      T_0 & 
      O \\ 
\noalign{\vspace{10pt}}
      T_1 & T_0
    \end{array}
  \right ] ^{-1} =
  \left[
    \begin{array}[c]{ccc}
      T_0^{-1} & 
      O\\  \noalign{\vspace{10pt}}
      -T_0^{-1}\, T_1\, T_0^{-1} & T_0^{-1}
    \end{array}
  \right ]
  \enspace .
\end{equation}
We only seek the first column of the matrix $T^{-1}$.
Its computation amounts to solving the same problem for the half-size
triangular Toeplitz matrix
$T_0$ and to multiplication of 
each of the 
%!!!
$\frac{\gamma}{2} \times \frac{\gamma}{2}$ Toeplitz matrices $T_1$
and $T_0^{-1}$ by a vector.
Let $TTI(s)$ and $TM(s)$ denote the arithmetic cost of
$s\times s$ triangular Toeplitz matrix inversion
 and multiplying an $s\times s$ Toeplitz matrix by a vector,
 respectively. Then the above analysis implies that
%!!!
$TTI(\gamma)\le TTI(\gamma/2)+2TM(\gamma/2)$. 
 Recursively apply
this bound to 
%!!!
$TTI(\gamma/2^g)$ for $g=1,2,\dots$,
 and deduce that 
%!!!
$TTI(\gamma)\le \sum_{g=1}^hTM(\gamma/2^g)$.
The following simple lemma (cf. \cite[equations (2.4.3) and (2.4.4)]{Pan01})
reduce Toep\-litz-by-vec\-tor multiplication to polynomial multiplication and the extraction
of a subvector of the coefficient vector of the product,
thus implying that $TM(s)\le cs\lg s$
for a constant $c$ and consequently 
%!!!
$TTI(\gamma)< 2c\gamma\lg \gamma$. 

\begin{lemma} \label{lem:veceq}
  The vector equation
  {\scriptsize
  \begin{equation}\label{eq7}
    \begin{pmatrix}u_0&&O\\ \vdots&\ddots&\\ \vdots&\ddots&u_0\\ u_m&\ddots&\vdots\\ &\ddots&\vdots\\ O&&u_m\end{pmatrix}\begin{pmatrix}v_0\\ \vdots\\ v_n\end{pmatrix}=\begin{pmatrix}p_0\\ \vdots\\ \vdots\\ p_m\\ \vdots\\ p_{m+n}\end{pmatrix}
  \end{equation}
  is equivalent to the polynomial equation
  \begin{equation}\label{eq8}
    \left(\sum_{i=0}^mu_ix^i\right)\left(\sum_{i=0}^nv_ix^i\right)=\sum_{i=0}^{m+n}p_ix^i.
  \end{equation}
  }
\end{lemma}

% {\bf Elias, we can drop the paragraph below or place it elsewhere}
% %ET I dropped

% Equations (\ref{eq:T-lin-eq}),
% (\ref{eq7}), and (\ref{eq8}) imply the equivalence of 
% univariate polynomial multiplication and division 
% to some basic operations with Toeplitz matrices.
% The paper \cite{BiPa-joc-1986} shows the equivalence of
%  recursive application of equation 
% (\ref{eq:T-inv}) to the Sieveking--Kung algorithm for
% computing the reciprocal of a
% polynomial modulo a power. We  
% employ the matrix version of this algorithm
% to make it more transparent.

We wish to estimate the Boolean (rather than arithmetic) cost of inverting 
a triangular Toeplitz matrix $T$ and then extend this to 
the Boolean cost bound of
computing the vector $T^{-1}{\bf s}$ 
and of polynomial division.
So next we assume that
 the input polynomials are known up to some precision
$2^{-\lambda}$ and employ the above 
reduction of the problem to 
recursive (approximate) polynomial multiplications.

To study the Boolean complexity of this procedure, we need
the following corollary, which is a direct consequence of Lemma~\ref{lem:a-poly-mult}
and the inequality $\lg(2d+1) \leq 2 + \lg{d}$.
\begin{corollary}[Bounds for the  product $P_0^2 \, P_1$]
  \label{cor:p02p1-mul}
  Let a polynomial $P_0\in \CC[x]$ 
  have a degree $d$, let its coefficients 
  be known up to a 
  precision
  $2^{-\nu}$, and 
  let $\norm{P_0}_{\infty} \leq
  2^{\tau_0}$.  Similarly, let $P_1\in \CC[x]$ 
  have the degree $2d$, let its coefficients 
  be known up to a 
  precision
  $2^{-\nu}$, and
  let
  $\norm{P_1}_{\infty} \leq 2^{\tau_1}$.  
  Then the polynomial $P = P_0^2 P_1$ has degree $4d$,
  its coefficients are known up to the precision 
  % 
  % $2^{-\nu +  8\tau_0 + 2\tau_1 + 14\lg{d} + 42}
  % %% V %% NO 
  % $,
  % and $\norm{P_0^2 P_1}_{\infty} \leq 2^{2\tau_0 + \tau_1 + 2\lg{d} + 3}
  % %% V %% NO: I put the estimates below.
  % 
  $2^{-\nu +  8\tau_0 + 2\tau_1 + 15\lg{d} + 40}$,
  and $\norm{P_0^2 P_1}_{\infty} \leq 2^{2\tau_0 + \tau_1 + 6\lg{d} + 8}$.
\end{corollary}

% \medskip
% {\bf Elias, only from here we differ from JSC.}
% \medskip

The following lemma is a normalized version of Lemma~4.4 in \cite{kirrinnis-joc-1998}.
\begin{lemma}
  \label{lem:quo-rem-bd}
  Let $F, G \in \CC[x]$ such that $\deg(F) = m \geq n = \deg(G) \geq 1$,
  let
  $2^{\rho}$ be an upper bound on the magnitude of roots of $G$,
  and let
  $F = G Q + R$ with $\deg(Q) = m - n$ and $\deg(R) = n - 1$.
Then 
  \[ \norm{Q}_{\infty} \leq 2^{m + \lg{m} + m \rho} \norm{F}_\infty 
  \ \text{ and } \ 
  \norm{R}_{\infty} \leq 2^{m + n + \lg{m} + m \rho} \norm{F}_\infty  \enspace .\]
\end{lemma}
\begin{proof}
  To bring the roots inside the unit circle,
  transform the polynomials by scaling the variable $x$ as follows, 
  $f(x) = F(x \,2^{\rho})$,
  $g(x) = G(x \,2^{\rho})$,
  $q(x) = Q(x \,2^{\rho})$, and
  $r(x) = R(x \,2^{\rho})$.
  Now apply 
  \cite[Lemma~4.4]{kirrinnis-joc-1998} to the equation $f = g q + r$ 
  to obtain
  $\norm{q}_{\infty} \leq \norm{q}_1 \leq 2^{m-1} \norm{f}_1 \leq 2^{m + \lg{m}} \norm{f}_\infty$
  and 
  $\norm{r}_{\infty} \leq \norm{r}_1 \leq \frac{3}{4} 2^{m+n} \norm{f}_1 \leq 2^{m+n + \lg{m}} \norm{f}_\infty$.
  % \[ \norm{q}_{\infty} \leq \norm{q}_1 \leq 2^{m-1} \norm{f}_1 \leq 2^{m + \lg{m}} \norm{f}_\infty
  % \quad \text{ and } \quad
  % \norm{r}_{\infty} \leq \norm{r}_1 \leq \frac{3}{4} 2^{m+n} \norm{f}_1 \leq 2^{m+n + \lg{m}} \norm{f}_\infty .\]
    
  Combine these inequalities 
  with the equation   $\norm{f}_{\infty} = 2^{m \rho} \norm{F}_{\infty}$
and the inequalities $ \norm{Q}_{\infty} \leq  \norm{q}_{\infty}$ and
$ \norm{R}_{\infty} \leq \norm{r}_{\infty}$ 
  to deduce the claimed bounds.
%y  $ \norm{Q}_{\infty} \leq \norm{q}_{\infty} 
  %y \leq 2^{m + \lg{m}} \norm{f}_\infty 
%y  \leq 2^{m + \lg{m} + m \rho} \norm{F}_\infty$
%y  and
%y  $ \norm{R}_{\infty} \leq \norm{r}_{\infty} 
%y  \leq 2^{m + n + \lg{m}} \norm{f}_\infty 
 %y \leq 2^{m + n + \lg{m} + m \rho} \norm{F}_\infty~. $
\end{proof}

We will 
estimate by induction the cost 
of inverting the matrix $T$, by using Eq.~(\ref{eq:T-inv}) recursively.
The proof of the following lemma could be found in the Appendix.
\begin{lemma}
  \label{lem:T-inv}
  Let $n+1=2^k$
  for a positive integer $k$
  and let $T$ be a lower 
  triangular Toeplitz $(n+1)\times (n+1)$ matrix of. Eq.~(\ref{eq:T-lin-eq}),
  having ones on the diagonal. Let its subdiagonal entries be
  complex numbers
  of magnitude at most $2^{\tau}$ known up to a precision
  $2^{-\lambda}$. 
  Let $2^{\rho}$ be an upper bounds on the magnitude of the roots of the 
  univariate polynomial 
  $t(x)$ associated with $T$.
  Write $T^{-1}=(T^{-1}_{i,j})_{i,j=0}^{n}$.
  Then 
  \[
  \max_{i,j} \abs{T^{-1}_{i,j}} \leq   2^{(\rho+1)n + \lg(n)+1} 
  \enspace .
  \]
% %!!!
% \medskip
% %{\bf Elias, it is actually $2^{(\rho+1)(n+1)+k}$, but I am not sure how to proceed. We can include precise expressions or can %say that for simplicity we replace $n+1$ by $n$ . (I have written such a sentence later). Similarly, we meet the same problem %is in the sequel.}
% \medskip

  Furthermore, to compute the
  entries of $T^{-1}$ up to the
  precision of $\ell$ bits, 
  that is to compute a matrix $\wt{T}^{-1}=(\wt{T}^{-1}_{i,j})_{i,j=0}^{n}$ 
  such that 
  %\begin{displaymath}
  $  \max_{i,j} \abs{T^{-1}_{i,j} - \wt{T}^{-1}_{i,j}} \leq 2^{-\ell} $,
  %  \enspace ,
  %\end{displaymath}
  it is sufficient to know the entries of $T$ up to
  the precision of

  $\ell + 10\tau \lg{n} + 70\lg^2{n} + 8(\rho+1)n\lg{n}$
%!!!
%{\bf (Elias, we can rewrite this as 
%$\ell + 10\tau k + 70k^2 + 8(\rho+1)(n+1)k$
 %or we can just say that $n+1$ is replaced by $n$ for simplicity)}
%y

  or $\OO(\ell + (\tau + \lg{n} +n \rho)\lg{n}) = \sO(\ell + \tau + n\rho)$ bits.

  The computation of $\wt{T}^{-1}$ costs
  $\OB( n \, \lg^{2}(n) \, \mu(\ell + (\tau + \lg{n} +n \rho)\lg{n}))$ 
  or $\sOB(n \ell + n\tau + n^2 \rho)$.
\end{lemma}
%%%
As usual in 
estimating the complexity of approximate division
we assume that $m = 2n$ to simplify 
our presentation.
Recall that $s(x) = t(x)\, q(x) + r(x)$.

\begin{theorem}
  \label{thm:approx-poly-div}
  Assume $s, t \in \CC[x]$ of degree at most $2n$ and $n$,
  such that   $\norm{s}_{\infty} \leq 2^{\tau_1}$, 
  $\norm{t}_{\infty} \leq 2^{\tau_2}$,
  and $2^{\rho}$ is an upper bound on the magnitude of the coefficients of $t(x)$.
  Assume that we know the coefficients of $s$ and $t$ up to 
  a precision $\lambda$,
  that is that the input includes two polynomials $\wt{s}$ and $\wt{t}$ such that 
  $\norm{s -\widetilde s}_{\infty} \leq 2^{-\lambda}$
  and $\norm{t -\widetilde t}_{\infty} \leq 2^{-\lambda}$,
  where  
  $\lambda = \ell + \tau_1 + 12\tau_2 \lg{n} + 80\lg^2{n} + 10(\rho+1)n\lg{n} + 30$
  or $\lambda = \OO(\ell + \tau_1 + \tau_2 \lg{n} +  n\, \rho \,\lg{n})$.
  Let $q$ denote the quotient and let $r$ denote  remainder
  of the division
  of the polynomials $s$ by $t$, that is $s = t \cdot q + r$ where $\deg r<\deg t$.
  
  Then we can compute in 
  $ \OB( n \, \lg^{2}(n) \, \mu(\ell + \tau_1 + (\tau_2 +n \rho)\lg{n})) $
  or $\sOB(n \ell + n\tau_1 + n\tau_2 + n^2 \rho)$
  two polynomials  $\wt{q}$ and $\wt{r}$ such that
  $ \norm{q - \wt{q}}_{\infty}\leq 2^{-\ell}$ and 
  $ \norm{r - \wt{r}}_{\infty}\leq 2^{-\ell}$,
  $\norm{q}_{\infty} \leq 2^{n + \lg{n} + 1 + n \rho + \tau_1}$  
  and $\norm{r}_{\infty} \leq 2^{ 3n + \lg{n} + 1 + n \rho + \tau_1}$.
\end{theorem}
\begin{proof}
  We compute the coefficients of $q(x)$ using Eq.~(\ref{eq:T-lin-eq}),
  ie $q = T^{-1} \, s$.  Each coefficient of the polynomial 
  $q$ comes as  the inner
  product of two vectors, ie $q_i = \sum_{j=0}^{n} T_{i,j}^{-1} s_j$.
  
  From Lemma~\ref{lem:T-inv} we know that
  ${\lg\abs{T_{i,j}^{-1}}} \leq n(\rho + 1) + \lg{n} + 1 = N$
  and $\lg{\abs{T_{i,j}^{-1} - \wt{T_{i,j}^{-1}}}} \leq 
  -\lambda + l_2$ for $l_2=10\tau_2 \lg{n} + 70\lg^2{n} + 8(\rho+1)n\lg{n}$. 
  
  For the coefficients of the polynomials $s=\sum_{j=0}^{2n}s_jx^j$
  and  $t=\sum_{j=0}^{n}t_jx^j$, we have assumed the following bounds, 
  $\lg\abs{s_j} \leq \tau_1$, $\lg{\abs{s_j - \wt{s}_j}} \leq -\lambda$,
  $\lg\abs{t_j} \leq \tau_2$, $\lg{\abs{t_j - \wt{t}_j}} \leq -\lambda$,
  $\lg \abs{T_{i,j}^{-1} s_j} \leq \tau_1 + N $,
  and
  $\lg \abs{T_{i,j}^{-1} s_j - \wt{T}_{i,j}^{-1} \wt{s}_j} \leq -\lambda + \ell_2 + \tau_1$
  for all $i$ and $j$. Therefore
  \[
  \begin{aligned}
    \lg \norm{q - \wt{q}}_{\infty} 
    & \leq \lg \abs{\sum_{j} T_{i,j}^{-1} s_j - \sum_{j} \wt{T}_{i,j}^{-1} \wt{s}_j} 
    \leq -\lambda + \ell_2 + \tau_1  + \lg{n}  \\
    & \leq -\lambda + 10\tau_2 \lg{n} + 70\lg^2{n} + 8(\rho+1)n\lg{n} + \tau_1 + \lg{n}
    \enspace .
  \end{aligned}
  \]

To compute 
the remainder we
apply the formula $r(x) = s(x) - t(x) q(x)$.
It involves an  approximate polynomial multiplication and a subtraction. 
For the former we use Lemma~\ref{lem:a-poly-mult}
and obtain the inequality
$\lg\norm{t \, q - \wt{t} \wt{q}}_{\infty} \
\leq -\lambda + 2\tau_2 + 6\lg{n} + 26 + \ell_2 + 2 N$.

Let us also cover the impact of the subtraction. 
After some calculations and
simplifications that make 
the bounds less scary (albeit less accurate wrt the 
constant involved), 
we obtain
{\scriptsize
\[
\begin{aligned}
  \lg \norm{r - \wt{r}}_{\infty}  
  & \leq  -\lambda + \tau_1 + 2\tau_2 + 6\lg{n} + 26 + \ell_2 + 2 N \\
  & \leq  -\lambda + \tau_1 + 2\tau_2 + 6\lg{n} + 26 + 10\tau_2 \lg{n} \\
  & \ \quad+ 70\lg^2{n} + 8(\rho+1)n\lg{n}
  + 2( n(\rho + 1) + \lg{n} + 1) \\
  & \leq -\lambda + \tau_1 + 12\tau_2 \lg{n} + 80\lg^2{n} + 10(\rho+1)n\lg{n} + 30
  \enspace .
\end{aligned}
\]
}

By using Lemma \ref{lem:quo-rem-bd} we bound the
norms of the
quotient and the remainder as follows:
$
\lg \norm{r}_{\infty} \leq 3n + \lg{n} + 1 + n \rho + \tau_1
\quad \text{ and } \quad
\lg \norm{q}_{\infty} \leq n + \lg{n} + 1 + n \rho + \tau_1
\enspace .
$

The maximum number of bits that we need to compute with is 
$\ell + \tau_1 + 12\tau_2 \lg{n} + 80\lg^2{n} + 10(\rho+1)n\lg{n} + 30$
or $\OO(\ell + \tau_1 + \tau_2 \lg{n} + \lg^2{n} +  n\, \rho \,\lg{n})$.

The complexity of computing $\wt{T}^{-1}_{i,j}$ is 
$\OB( n \, \lg^{2}(n) \,  \mu(\ell + \tau_1 + \tau_2 \lg{n} + \lg^2{n} +  n\, \rho \,\lg{n}))$
or $\sOB(n \ell + n\tau_1 + n\tau_2 + n^2 \rho)$.

According to Lemma~\ref{lem:a-poly-mult} the complexity of 
computing the product $\wt{t} \, \wt{q}$
is
$\OB( n \lg(n) \, \mu(\ell + \tau_1 + \tau_2 \lg{n} + \lg^2{n} +  n\, \rho \,\lg{n}))$
or $\sOB(n \ell + n\tau_1 + n\tau_2 + n^2 \rho)$.
\end{proof}

\begin{remark}
  \label{rem:on-rho}
  We can eliminate the dependence 
 of the bounds of Theorem \ref{thm:approx-poly-div} on $\tau_2$
  by applying Vieta's formulae and the
  following inequality,
  $\abs{t_k} \leq {n \choose k} (2^{\rho})^{k} \leq 2^{2n + n\rho}$,
  where $t_k$ is the $k$-th coefficient of $t(x)$.
  In this way, after some further simplifications, 
  the required precision is
  $\ell + \tau_1 + 150(\rho+1)n\lg{n}$
  and the complexity bound becomes
  $ \OB( n \, \lg^{2}(n) \, \mu(\ell + \tau_1 + \rho \lg{n})) $
  or $\sOB(n \ell + n\tau_1 + n^2 \rho)$.
\end{remark}

\section{Multipoint polynomial evaluation}\label{smpe}

\begin{problem}\label{prob3.1.1}
{\bf Multipoint polynomial evaluation.} Given the coefficients of a polynomial 
$p(x)=\sum_{i=0}^{n-1}p_ix^i$ and a set of knots  $t_0,\ldots,t_{n-1}$, compute the values 
$r_0=
%%V
p(t_0),\ldots,r_{n-1}=p(t_{n-1})$ or equivalently
compute the vector ${\bf r}=V{\bf p}$ where ${\bf r}=(r_i)_{i=0}^{n-1}$,
${\bf p}=(p_i)_{i=0}^{n-1}$, and $V=(x_i^j)_{i,j=0}^{n-1}$. 
\end{problem}

%Problem \ref{prob3.1.1} can be stated equivalently as follows.

%\begin{problem}\label{prob3.1.2}
%{\bf Vandermonde-by-vector product.} Given a Vandermonde matrix $V=V({\bf t})=(t_i^j)_{i,j=0}^{n-1}$ and a vector ${\bf v}=(v_j)_{j=0}^{n-1}$, compute the vector ${\bf r}=(r_i)_{i=0}^{n-1}=V{\bf v}$.
%\end{problem}

In the case where the knots $t_i=\omega^{i}$ are the $n$-th roots of 1 for all $i$,
 $\omega=\exp(2\pi \sqrt {-1})$/n, and 
$V=\Omega=(\omega^{ij})_{i,j=0}^{n-1}$, 
Problem \ref{prob3.1.1} turns into the problem of the DFT$(\bf v)$ computation. 

\paragraph*{Solution:} The Moenck--Borodin algorithm of \cite{MB72}
solves Problem \ref{prob3.1.1} in $O(M(n)\log n)$ ops for $M(n)$ in (2.4.1), (2.4.2) based on the 
two following simple observations.

\begin{fact}\label{fact3.1.3}
  $p(a)=p(x)\bmod(x-a)$ for any polynomial 
  $p(x)$ and any scalar $a$.
\end{fact}

\begin{fact}\label{fact3.1.4}
  $w(x)\bmod p(x)= (w(x)\bmod(u(x)p(x)))\bmod p(x)$ 
  for any triple of polynomials $u(x)$, 
  $p(x)$, and $w(x)$.
\end{fact}

{\bf Algorithm 4: the Moenck--Borodin algorithm for multipoint polynomial evaluation}.

\paragraph*{\sc initialization:} Write $k=\lceil\log_2n\rceil$, $m_j^{(0)}=x-x_j$, $j=0,1,\ldots,n-1$; $m_j^{(0)}=1$ for $j=n,\ldots,2^k-1$ (that is, pad the set of the moduli $m_j^{(0)}=x-x_j$ with ones, to make up a total of $2^k$ moduli). Write $r_0^{(k)}=p(x)$.

\paragraph*{\sc Computation:} 
\begin{enumerate}
\item {\em Fan-in process} (see Figure~\ref{fig:fan-in}). 
 Compute recursively the ``supermoduli" $m_j^{(h+1)}=m_{2j}^{(h)}m_{2j+1}^{(h)}$, $j=0,1,\ldots,2^{k-h}-1$; $h=0,1,\ldots,k-2$.
\item {\em Fan-out process} (see Figure~\ref{fig:fan-out}).
  Compute recursively the remainders $r_j^{(h)}=r_{\lfloor j/2\rfloor}^{(h+1)}\bmod m_j^{(h)}$, \\
  $j=0,1,\ldots,\min\{n,\lceil n/2^h\rceil-1\}$; $h=k-1,k-2,\ldots,0$.
\end{enumerate}

\paragraph*{\sc Output:} $p(x_i)=r_i^{(0)}$, $i=0,1,\ldots,n-1$.

Let us include a brief outline of the analysis of the algorithm
(cf.~\cite{MB72}).  To prove its {\em correctness}, first apply
Fact~\ref{fact3.1.4} recursively to obtain that $r_j^{(h)}=v(x)\bmod
m_j^{(h)}$ for all $j$ and $h$. Now, correctness of the output
$p(x_i)=r_i^{(0)}$ follows from Fact~\ref{fact3.1.3}.

To estimate the {\em computational cost} of the algorithm, represent
its two stages by the same binary tree (see Figures 3.1 and 3.2),
whose nodes are the ``supermoduli" $m_j^{(h)}$ at the {\em fan-in}
stage 1, but turn into the remainders $r_j^{(h)}$ at the {\em fan-out}
stage 2.

At each level $h$ of the tree, the algorithm computes $2^{k-h}$ products of pairs of polynomials of degree $2^h$ at stage 1 and $2^{k-h}$ remainders of the division of polynomials of degree of at most $2^{h+1}$ by ``supermoduli" of degree $2^h$. Each time multiplication/division uses $O(M(2^h))$ ops for $M(n)$ in (2.4.1), (2.4.2). 
%!!!
So we use 
%!!! a total of 
$O(2^{k-h}M(2^h))$ ops at the $h$-th level 
%!!!,
 and $O(\sum_{h=0}^{k-1}2^{k-h}M(2^h))=O(M(2^k)k)$ ops at all levels. Recall that $n\le 2^k<2n$ and obtain the claimed bound of $O(M(n)\log n)$ ops. \qed

\begin{remark}\label{rem3.1.6}
The fan-in computation at stage 1 depends only on the set
$\{t_0,\ldots,t_{n-1}\}$ and can be viewed as (cost-free) {\em preprocessing} if the 
knot set is fixed and only the polynomial 
$p(x)$ varies. Similar observations hold for the solution of many other problems in this chapter.
\end{remark}

\begin{remark}\label{rem3.1.11}
  Problem \ref{prob3.1.1} and its solution algorithms are immediately
  extended to the case where we have $m$ points $t_0,\ldots,t_{m-1}$
  for $m>n$ or $m<n$. The solution requires $O(E(l)r/l)$ ops provided
  $l=\min\{m,n\}$, $r=\max\{m,n\}$, and $E(l)$ ops are sufficient for
  the solution where $n=l$. $E(l)=O(M(l)\log l)$ for a general set
  $\{t_i \}$ but decreases to $O(M(l))$,
  where $t_i=at^{2i}+bt^i+c$ for
  fixed scalars $a,b,c$, and $t$ and for all $i$.
  This also leads to a similar improvement of 
  the estimates for 
  the Boolean complexity 
  \cite{ASU75}.
\end{remark}

% \bigskip
% A. V. Aho, V. Steiglitz, J. D. Ullman,
% Evaluation of Polynomials at Fixed Sets of Points,
% {\em SIAM Journal on Computing}, {\bf 4}, 533--539, 1975.
% \bigskip

\subsection{Boolean complexity estimates}

In the following two lemmata we present the bit complexity of the
fan-in and the fan-out process.  These results are of independent
interest.  We do not estimate the accuracy needed and the bit
complexity bound of the algorithm for multipoint evaluation 
because  in
Lemma~\ref{lem:rem-m-polys} we cover a more general  algorithm.
Multipoint evaluation is its special case.

\begin{lemma}[Complexity of Fan-in process]
  \label{lem:fan-in}
  Assume that we are given $n$ complex numbers $x_i$
  known up to a precision $\lambda = \ell + (4n-4)\tau + 32n - (\lg{n} +5)^2 - 7$, that is 
  $\abs{x_i - \wt{x}_i} \leq 2^{-\lambda}$,
  and  that 
  $\abs{x_i} \leq 2^{\tau}$ for a positive integer $\tau$.
  At the cost
  $\sOB( n \lg^{2}{n} \, \mu(\ell +n\tau + \lg{n}) ) $
  the Fan-in process of the Moenck--Borodin algorithm
  approximates the ``supermoduli" $\wt{m}_j^{(i)}$ within the bounds
  $\norm{m_j^{(i)} - \wt{m}_j^{(i)}}_{\infty} \leq 2^{-\ell}$ for all $i$ and $j$.
  Moreover,  $\lg \norm{m_{j}^{(i)}}_{\infty} \leq n\tau + 8n - 2\lg{n} - 8$
  for all $i$ and $j$.
\end{lemma}
\begin{proof}
  Assume that $n=2^{k}$. The proof is by induction on $k$.
  Write $m_i^{(0)} = x - x_i$ and $\wt{m}_i^{(0)} = x - \wt{x}_i.$ 
  Wlog we provide the estimates just in the case where $j=0$.  
  
  Consider the case where $k=1$.
  Apply Lemma~\ref{lem:a-poly-mult} for $A=m_0^{(0)}$ and $B=m_1^{(0)}$.
  Verify that $\lg \norm{m_0^{(1)}}_{\infty} \leq 2\tau \leq 2\tau + 6$ 
  and
  %% V 
  $\lg \norm{m_0^{(1)} - \wt{m}_0^{(1)}}_{\infty} \leq -\lambda + 4\tau + 14.2$.
  This proves the induction basis.
  
  Now assume that the claimed 
  bounds hold for $k-1$, that is 
  $\deg(m_i^{(k-1)}) = 2^{k-1}$,
  $\lg \norm{ m_i^{(k-1)} }_{\infty} \leq 2^{k-1}\tau + 2^{k+2} - 2k - 6$,
  and 
  $\lg \norm{m_i^{(k-1)} - \wt{m}_i^{(k-1)}}_{\infty} 
  \leq -\lambda + (2^{k+1}-4)\tau + 2^{k+4} -(k+4)^2 -7$
  for $i \in \{0,1\}$.

  Since $m_0^{(k)} = m_0^{(k-1)}\, m_1^{(k-1)}$,
  it follows that 
  $\deg(m_0^{(k)}) = 2^{k}$.
 By applying Lemma~\ref{lem:a-poly-mult} we deduce that
  \[ \lg \norm{ m_i^{(k)} }_{\infty} = 
  \lg \norm{ m_0^{(k-1)} \, m_1^{(k-1)} }_{\infty}
  \leq 2^{k}\tau + 2^{k+3} - 2k - 8
  \]
  and 
%\medskip
%
%!!! {\bf ET: I think that I fixed all the formulas}
%
%!!! {\bf Elias, I am not sure about $(2^{k+2}-4)\tau$. Let us drop superscripts $(k-1)$. Then
%$|m_0m_1-\tilde m_0\tilde m_1|= |m_0(m_1-\tilde m_1)+(m_0-\tilde m_0)\tilde m_1|$
%and we get an upper bound $2|m_i| |m_j-\tilde m_j|$. Substitute assumed bounds on
%$2|m_i|$ and $|m_j-\tilde m_j|$ and obtain $(2^{k-1}+2^{k+1})\tau+.......$. If this is right,
%then some fixes are needed.}
% \medskip
%
  \[
  \begin{array}{ll}
  \lg \norm{m_0^{(k)} - \wt{m}_0^{(k)}}_{\infty} 
  & =  
  \lg \norm{m_0^{(k-1)} \, m_1^{(k-1)} - \wt{m}_0^{(k-1)} \, \wt{m}_1^{(k-1)} }_{\infty}  \\
  & \leq -\lambda + (2^{k+2}-4)\tau + 2^{k+5} -(k+5)^2 - 7
  \end{array}
  \]
  as claimed.
 
  To estimate the overall complexity, 
  note that at the $h$th level of the tree 
  we  perform $n/2^{h}$ multiplications of polynomials of degrees at most  $2^{h-1}$
  for $h=2,\dots,k-1$.
  We can assume that we 
  perform all the computation with precision $\OO(\ell + n\tau + \lg{n})$,
  and so the overall cost of the algorithm is
  $\sum_k \frac{n}{2^k} \sOB(2^{k-1} \lg{2^{k-1}} \mu(\ell +n\tau + \lg{n}) ) =
  \sOB( n \lg^{2}{n} \, \mu(\ell +n\tau + \lg{n}) ) $.  
\end{proof}

\begin{lemma}[Complexity of Fan-out process]
  \label{lem:fan-out} 
  Let $v(x) \in \CC[x]$ of degree $n-1$ and $\normi{v} \leq 2^{\tau_1}$,
  and let $\wt{v}$ be a $\lambda$-approximation.
  Let $m_j^{(k)}$ be the supermoduli of the {\em fan-in} process and 
  $\wt{m}_{j}^{(k)}$ their $\lambda$-approximations.
  
 %%V  If   
We can compute an $\ell$-approximation of the
  {\em fan-out} process in 
  $\OB( n \, \lg^{2}{n} \, \mu( \ell + \tau_1 \lg{n} + \rho \, {n} \lg{n}))$
provided that $\lambda = \ell + 2\tau_1 \lg{n} + 300(\rho+1)n\lg{n}$.
\end{lemma}
\begin{proof}
  %%V 
  We keep assuming  for simplicity that $n = 2^k$
%%V
and proceed as in the proof of
  Lemma~\ref{lem:fan-in}.

  Recall that  $\abs{x_i} \leq 2^{\rho}$ for all the 
%%V
subscripts $i$,
  and so $2^{\rho}$ bounds the roots of all polynomials $m_j^{(k)}$.
  We can prove by induction, by
%%V
  using the bounds of Theorem~\ref{thm:approx-poly-div} 
  and the simplifications of Remark~\ref{rem:on-rho},
  that the 
%%V needed 
precision of
  $\lambda = \ell + 2\tau_1 \lg{n} + 300(\rho+1)n\lg{n}$ bits 
is sufficient.

  At the $h$th step of the algorithm, for each $h$, 
%%V
 we perform $2^{h}$
%%V
 approximate polynomial divisions 
  of polynomials of degree $\frac{n}{2^h}$
%%V
 using Theorem~\ref{thm:approx-poly-div}.
  We 
 assume 
%%V the 
 performing all the operations 
%%V
with the maximum 
%%V needed 
precision, 
  and bound the overall complexity by 
%%V  overall complexity is
  \[
  \sum_{h=0}^{\lg{n}} 2^{h}\, \OB( \frac{n}{2^{k}} \, (\lg{\frac{n}{2^h}})^2 \,
                   \mu( \ell + \tau_1 \lg{n} + \rho \, {n} \lg{n}))
                   \]
%%V (replace k by h 3 times)
                   \vspace{-10pt}
                   \[
  =
  \OB( n \, \lg^{2}{n} \, \mu( \ell + \tau_1 \lg{n} + \rho \, {n} \lg{n}))
  \]
\end{proof}

% We choose as inout precision 
% $\lambda = \ell + (4n-4)\tau_2 + 16n + 20 +  2\tau_1 \lg{n} + 300(\rho+1)n\lg{n}$
% or 
% $\lambda = \ell + 2\tau_1 \lg{n} + 350(\rho+1)n\lg{n}$.
% The constant 350 in the definition of $\lambda$ is by no means the optimal one.

% By combining the two previous lemmata we arrive at the following theorem:
% %%
% \begin{theorem}[Complexity of multipoint evaluation]
%   Let $v(x) \in \CC[x]$ of degree $n-1$ and $\normi{v} \leq 2^{\tau_1}$,
%   and let $\wt{v}$ be a $\lambda$-approximation.
%   Let $t_0, \dots, t_{n-1}$ be a set of knots, such that 
%   $t_i \in \CC$, $\abs{t_i} \leq 2^{\tau_2}$ 
%   and $\wt{t}_i$ be their $\lambda$ approximations.

%   If 
%   $\lambda = \ell + 2\tau_1 \lg{n} + 350(\rho+1)n\lg{n}$,
%   then we can compute an $\ell$-approximation
%   of the values $v(t_0), \dots, v(t_{n-1})$ in 
%   \[ \OB( n \cdot \lg^{2}{n} \cdot \mu( \ell + \tau_1 \lg{n} + \tau_2 \, {n} \lg{n})) \]
%   or
%   \[ \sOB( n ( \ell + \tau_1 + n \tau_2 )) \]
% \end{theorem}

%!!!! 	\newpage

\section{Bounds on 
%!!!
the complexity of basic algorithms}

\begin{lemma}[Multiplication of $m$ polynomials]
  \label{lem:a-mul-m-polys}
  Let $P_j \in \CC[x]$ of degree $n$ and $\normi{P_j} \leq 2^{\tau}$.
  Let $\wt{P}_j$ be $\lambda$-approximation of $P_j$ 
  with $\lambda = \ell + (4m-4)\tau + (4m+2\lg{m}-4)\lg{n} + 32m$,
  where $1 \leq j \leq m$.
  We can compute $\prod_j\wt{P}_j$ such that 
  $\normi{ \prod_j P_j - \prod_j \wt{P}_j } \leq 2^{-\ell}$
  in $\OB( m \, n \, \lg{m} \lg(m \, n) \, \mu(\lambda))$
  or
  $\sOB(m \, n \,(\ell + m \tau))$.
  Moreover, 
  $\lg\normi{\prod_j P_j} \leq m\tau + (m-1)\lg{n} + 4m - \lg{m} - 4$.
\end{lemma}
\begin{proof}
  The algorithm is similar to the Fan-in process of Moenck-Borodin algorithm.
  Let 
  $p_j^{(0)} = P_j$ and compute recursively the polynomials
  $p_j^{(h+1)} = p_{2j}^{(h)} p_{2j+1}^{(h)}$, 
  for $0 \leq j \leq 2^{k-h}$, $h=0, \dots, k-2$.
  Let $m = 2^{h}$.
  We prove the bounds on the infinite norm and the approximation using
  induction on $h$.
  
  For $h=1$, we compute the polynomials $p_j^{(1)} = p_{2j}^{(0)} p_{2j+1}^{(0)}$.
  Wlog assume that $j=0$. 
  Then $\lg \normi{p_{j}^{(0)}} \leq {\tau}$
  and $\lg\normi{p_j^{(0)} - \wt{p}_j^{(0)}} \leq -\lambda$.
  Apply Lemma~\ref{lem:a-poly-mult} (and Remark~\ref{rem:a-poly-mult})
  for $K=n$ and $\tau_1=\tau_2=\tau$
  to deduce that  
  \[
  \lg\normi{p_0^{(1)}} \leq 2\tau + \lg{n} + 3
  \enspace ,
  \]
  which agrees with our formula,
  and
  \[
  \lg\normi{p_j^{(1)} - \wt{p}_j^{(1)}} \leq -\lambda + 4\tau + 
  5.1\lg{n} + 4 \leq  -\lambda + 4\tau + 6\lg{n} + 64 
  \enspace ,
  \]
  where the right hand-side represents the claimed bound for $k=1$.

  Assume 
%!!! that 
the claimed bounds
%!!!  hold 
for $h-1$, that is
  \[
  \lg\normi{p_j^{(h-1)}} \leq 2^{h-1} \tau + (2^{h-1}-1)\lg{n} + 4\cdot2^{h-1} - \lg{2^{h-1}} - 4
  \]
  and
  $
  \lg\normi{p_j^{(h-1)} - \wt{p}_j^{(h-1)}} 
  \leq -\lambda + (4 \cdot 2^{h-1}-4)\tau + (4 \cdot2^{h-1}+2\lg{2^{h-1}}-4)\lg{n} + 32\cdot2^{h-1}
  $
  for $j=0,1$.  By applying Lemma~\ref{lem:a-poly-mult}  
  for $2lg(K)\le \lg (n)$,
  deduce
  the following bounds,
  \[
  \lg\normi{p_j^{(h)}} \leq 2^{h} \tau + (2^{h}-1)\lg{n} + 4\cdot2^{h} - \lg{2^{h}} - 4,
  \]
  which agrees with 
  the claimed norm bound, and 
  $
  \lg\normi{p_j^{(h)} - \wt{p}_j^{(h)}} 
  \leq -\lambda + (4 \cdot 2^{h}-4)\tau + (4 \cdot2^{h}+2\lg{2^{h}}-4)\lg{n} + 24\cdot2^{h}+2h - 12
  $
  which is smaller than the claimed bound on the precision.

To estimate the overall complexity 
note that at each level, $h$, of the tree 
we have to perform $m/2^{h}$ multiplications of polynomials of degrees
at most 
$2^{h-1} \,n $.
We can assume that we 
perform all the computations with 
the precision 
$\lambda + (4m-4)\tau + (4m+2\lg{m}-4)\lg{n} + 32m$,
or 
$\OO(\ell + m\tau + m\lg{n})$,
and so the overall Boolean 
%!!!
cost of 
%!!!
performing the algorithm is
$
\sum_h \frac{m}{2^h} \OB( 2^{h-1} \, n  \lg{(2^{h-1} \, n)} \mu(\ell +n\tau + m\lg{n}) ) =
\OB( m\,n\, \lg{m} \lg(m\,n) \, \mu(\ell +n\tau + m\lg{n}) ) ~,
$
which concludes the proof.
\end{proof}

If the degrees of $P_j$ 
vary as $j$ varies,
then we can apply 
a more pedantic analysis based on 
Huffman trees, 
see
\cite{kirrinnis-joc-1998}.

The problem of computing (approximately) the sum of rational 
%!!! polynomial 
functions
reduces to the problem on multiplying polynomials,
which admits the same  asymptotic 
complexity bounds. 
To estimate the overhead 
constants,
we should also take into account 
the polynomial additions
involved.

We have the following lemma.
\begin{lemma}[Sum of rational functions]
  \label{lem:a-sum-rat-func}
  Suppose $P_j \in \CC[x]$ has degree $n$, 
  $Q_j \in \CC[x]$ has a smaller degree, 
  $\normi{P_j} \leq 2^{\tau_2}$, and $\normi{Q_j} \leq 2^{\tau_1}$.
  
  Assume $\lambda$-approximations of $P_j$ by  $\wt{P}_j$ and   
  of $Q_j$ by $\wt{Q}_j$
  where $\lambda = \ell + \tau_1 + (4m-4)\tau_2 + (5m+2\lg{m}-4)\lg{n} + 32m$
  and $1 \leq j \leq m$.

  Let $\frac{Q}{P} = \sum_j\frac{Q_j}{P_j}$. 
  We can compute 
  an $\ell$-approximation of $Q/P$,
  in $\OB( m \, n \, \lg{m} \lg(m \, n) \, \mu(\lambda))$
  or
  $\sOB(m \, n \,(\ell + +\tau_1 + m \tau_2))$.
  Moreover, 
  $\lg\normi{Q} \leq \tau_2 + (m-1)(\tau_2 + \lg{n}) + 5m - \lg{m} - 4$ and 
  $\lg\normi{P} \leq m\,\tau_2 + (m-1)\lg{n} + 4 m - \lg{m} - 4$.
\end{lemma}

\begin{lemma}[Modular representation]
  \label{lem:rem-m-polys}
  Let $F \in \CC[x]$ of degree $2 m n$ and $\normi{F} \leq 2^{\tau_1}$.
  Let $P_j \in \CC[x]$ of degree $n$, for $ 1 \leq j \leq m$.
  Moreover, $2^{\rho}$ be an upper bound on the magnitude of the roots of all $P_j$,
  for all $j$.
  Assume $\lambda$-approximations of $F$ by  $\wt{F}$ and   
  of $P_j$ by $\wt{P}_j$
  such that 
  $\normi{F - \wt{F}} \leq 2^{-\lambda}$
  and 
  $\normi{P_j - \wt{P}_j} \leq 2^{-\lambda}$.

  Furthermore assume that
  $\lambda= \ell+ \tau_1 \lg{m} + 60 \, n \, m (\rho + 3) \,lg(m\,n) + 60 \lg{m} \lg^2(m+n)$
  or $\lambda = \ell + \OO(\tau_1\lg{m} + m \, n \, \rho)$.
  Then we can compute 
  an $\ell$-approximation $\wt{F}_j$ of $F_j = F \mod P_j$
  such that $\normi{F_j - \wt{F}_j} \leq 2^{-\ell}$
  in 
  $$ \OB( m \, n \, \lg{n} \, \lg^{2}(m+n) \, \mu( \ell + \tau_1\lg{m} + m\,n\,\rho))$$
  or $\sOB(m \, n \,(\ell + \tau_1 + m\,n\,\rho))$.

  Moreover, 
  $\lg\normi{F_j} \leq \tau_1 + (\rho+1)\,m\,n + n +\lg(m\,n)$.
\end{lemma}

\begin{proof}
  First we perform the Fan-in process 
  with polynomials $P_j$ using the algorithm of Lemma~\ref{lem:a-mul-m-polys}.
  Assume that $\normi{P_j} \leq 2^{\tau_2}$
  and that we are given $\lambda_1$-approximations.
  Following Lemma~\ref{lem:a-mul-m-polys} we compute all the supermoduli,
  $P_j^{(i)}$ so that 
  \[
  \lg\normi{P_j^{(i)} - \wt{P}_j^{(i)}}
  \leq -\lambda_1 +  (4m-4)\tau_2 + (4m+2\lg{m}-4)\lg{n} + 32m~.
%!!!
  \]
  Remark~\ref{rem:on-rho} implies that
  $\tau_2 \leq 2n + n\rho$,
  and so
  $
  \lg\normi{P_j^{(i)} - \wt{P}_j^{(i)}}
  \leq -\lambda + (4m-4)(2n + \lg{n} + n\rho) + 2\lg{m}\lg{n} + 32m
  = -\lambda + \OO(m \,n\, \rho)~.
  $

  For computing $\ell$-approximations of $F_j = F \mod P_j$
  we mimic the procedure of the Fan-out process.
  This means that we apply repeatedly 
  Theorem~\ref{thm:approx-poly-div}, 
  which we can refine by following Remark~\ref{rem:on-rho}. 
  The bounds accumulate at each step, 
  and so 
  \[
  \begin{array}{ll}
  \lg\normi{F_j} & \leq \tau_1 + \sum_h 3\, n \, 2^{h} + n + h + 2^h n \rho + 1 \\
  & \leq (m n - n)(\rho +3) + m(m+n) .
  \end{array}
  \]

  We assume that we are given $\lambda_2$-approximation of $F$ 
  and all the supermoduli. 
  For the required precision we have
  {\scriptsize
  \[
  \begin{aligned}
  \lg\normi{F_j - \wt{F}_j} 
  & \leq -\lambda_2 + \sum_h 25(\rho+2)(h + \lg{n}) n 2^{h} + 80(h + \lg{n})^2 + \tau_1 + 30 \\
  & \leq -\lambda_2 + \tau_1 \lg{m} + 25 \, n \, m (\rho + 2) \,lg(m\,n) + 40 \lg{m} \lg^2(m+n) .
  \end{aligned} 
  \]
  }

  To ensure an $\ell$-approximation for $F_j$ 
  we require 
  $
  \lambda= \ell+ \tau_1 \lg{m} + 60 \, n \, m (\rho + 3) \,lg(m\,n) + 60 \lg{m} \lg^2(m+n)
  = \ell + \OO(\tau_1\lg{m} + m \, n \, \rho)
  $
  approximations of the input to ensure the validity of both the Fan-in and Fan-out process.

  We assume that we perform all the computations with maximum accuracy. 
  The complexity of computing the super-moduli is 
  $\OB( m \, n \, \lg{m} \lg(m \, n) \, \mu(\ell + \tau_1\lg{m} + m\,n\,\rho))$
  or
  $\sOB(m \, n \,(\ell + \tau_1 + m\,n\,\rho))$.

  For the complexity of the Fan-out process we proceed as follows.
  At each step, $h$, of the algorithm we perform $2^{h}$ approximate polynomial divisions 
  of polynomials of degree $\frac{m \,n}{2^h}$ using Theorem~\ref{thm:approx-poly-div}.
  The overall complexity is\\
  $
  \sum_{h=0}^{\lg{m}} 2^{h}\, \OB( \frac{m n}{2^{h}} \, (\lg{\frac{m n}{2^h}})^2 \,
                   \mu( \ell + \tau_1\lg{m} + m\,n\, \rho))$
                   \\
                   $= \OB( m \, n \, \lg{n} \, \lg^{2}(m+n) \, \mu( \ell + \tau_1\lg{m} + m\,n\,\rho))
  $
  or $\sOB(m \, n \,(\ell + \tau_1 + m\,n\,\rho))$.

  % in $\OB( m \, n \, \lg{m} \lg(m \, n) \, \mu(\ell + \tau_1\lg{m} + m\,n\,\rho))$
  % or
  % $\sOB(m \, n \,(\ell + \tau_1 + m\,n\,\rho))$.
  % $(m n - n)(\rho +3) + m(m+n)$
  % and $-\lambda + \tau_1 \lg{m} + 25 \, n \, m (\rho + 2) \,lg(m\,n)) + 40 \lg{m} \lg^2(m+n)$
  % Now we combine with the bounds of multiplications.
\end{proof}

%!!!! \newpage

\section{Lagrange Interpolation}

\begin{problem}\label{prob:Lagrange-interpolation} 
  {\bf Lagrange polynomial interpolation.} Given the 
  knot set (or
  vector) $\{x_i\}_{i=0}^{n-1}$ of $n$ distinct points
  $x_0,\ldots,x_{n-1}$ and the set (or vector) of values
  $\{y_i\}_{i=0}^{n-1}$, compute a set (or vector)
  $\{a_j\}_{j=0}^{n-1}$ such that $\sum_{j=0}^{n-1}a_jx_i^j=y_i$,
  $i=0,1,\ldots,n-1$, that is, recover the coefficients of a
  polynomial $A(X)=\sum_{j=0}^{n-1}a_j X^j$ from its values at $n$
  distinct points $x_0,\ldots,x_{n-1}$.
\end{problem}

We follow the approach presented in \cite[Section~3.3]{Pan01}, 
to which
we also refer for a detailed presentation.

\begin{lemma}
  \label{lem:lagrange}
  Let $\abs{x_i} \leq 2^{\tau_1}$, $\abs{y_i} \leq 2^{\tau_2}$, 
  and $\Delta_i(x) = \min_j\abs{x_i - x_j}$,
  for all $0 \leq i \leq n-1$.
Assume $\lambda$-approximations of $x_i$ and $y_i$, 
  where   
  $\lambda = \ell + 68 n (\tau_1+3) \lg{n} + 4 n \tau_2 - 6\lg\prod_i\Delta_i(x) + 50 n + 60\lg^{3}{n} + 20$
  or $\lambda = \ell + \OO(n\tau_1\lg{n} + n\tau_2 -\lg\prod_i\Delta_i(x) + \lg^{3}{n})$.
Then  we can compute 
  an $\ell$-approximation of the Lagrange polynomial interpolation
  in $\OB( n \, \lg^2{n} \, \mu(\lambda) )$
  or $\sOB(  n^2\tau_1 + n^2\tau_2 - n \lg\prod_i\Delta_i(x))) $.
\end{lemma}
\begin{proof}
  The input is given as 
  $\lambda$-approximations, where $\lambda$ is
  to be specified in the sequel.
  We track the loss of accuracy at each step of the algorithm.
  
  \begin{enumerate}[leftmargin=5pt,rightmargin=3pt]
  \item Compute $B(X) = \prod_{i}(X - x_i) = X^n + \sum_{k=0}^{n-1}b_k X^k$.
    
    For this task we apply Lemma~\ref{lem:fan-in}.
    In this case the infinite norm of $B$ is bounded as follows,
    $\lg\normi{B} \leq n \tau_1 + 8n - 2\lg{n} - 4$,
    and the computed approximation, $\wt{B}$, is such that  
    $\lg\normi{B~-~\wt{B}} \leq -\lambda_0 + (4n - 4)\tau_1 + 16 n + 20$.
    As by-product we can compute the ``supermoduli"
    $\prod_j(x-x_j)$ and then reuse them at stage L5.

\item Compute $B'(X)$.
  
  This operation increases
  the norm and the precision 
  bounds
  by a factor of $n$ in the worst case.
  That is 
  $\lg\normi{B'} \leq n\tau_1 + 8n - \lg{n} -4$
  and
  $\lg{\normi{B' - \wt{B'}}} \leq -\lambda_0 + (4n - 4)\tau_1 + 16 n + \lg{n}+ 20 = -\lambda_1$.

\item Evaluate $B'$ at all 
  points $x_i$.
 
  Perform this task using Lemma~\ref{lem:rem-m-polys}. 
  This lemma implies that 
  $\lg\abs{B'(x_i)} \leq (n-1)(\tau_1 + 3) + n(n+1)$.
  However, in this special case 
  we can decrease the bound
  as follows,
  $
  \abs{B'(x_i)} \leq \sum_{j}\normi{B'} \, \abs{x_i}^{n-1} 
  \leq \sum_{j} 2^{n\tau_1 + 8n -\lg{n} - 4} \, 2^{(n-1) \tau_1}
  \leq 2^{(2n-1)\tau_1 + 8n - 4}
  \enspace .
  $

  We achieve the accuracy 
  $\lg\abs{B'(x_i) - \wt{B'}(\wt{x}_i)} \leq -\lambda_1 + (n\tau_1 + n -\lg{n} - 1)\lg{n} 
  + 60n(\tau_1+3)\lg{n} + 60 \lg^3{n} = -\lambda_2~.$ 

\item Consider the rational functions
  $A_i(X) = \frac{A_{i,0}(X)}{A_{i,1}(X)} = \frac{y_i/B'(x_i)}{(X - x_i)}$.
  
  Deduce that $\lg{\normi{A_{i,1}}} \leq \tau_1$,
  and 
  so the approximation bound matches that of $x_i$.
  
  To compute the relevant quantities of the numerator(s)
  we need a lower bound for  $B'(x_i)$, for all $i$. 
  We notice that 
  $B'(X) = \sum_{i=1}^{n}\prod_{j \not= i}(X - x_j)$.
  Thus $B'(x_i) = \prod_{j \not= i}(x_i - x_j)$ and so
  $\abs{B'(x_i)} \geq  \prod_{j \not= i}\Delta_j(x)$,
  and
  $\lg \normi{A_{i,0}} \leq  \tau_2 - \lg\prod_{j\not= i}\Delta_j(x)
  \leq  \tau_2 - \lg\prod_{j}\Delta_j(x)$.
  
  For computing an approximation of the denominator we rely on (complex) interval arithmetic.
  That is
  $\abs{A_{i,0} - \wt{A}_{i,0}} \leq \wid([A_{i,0}]) = \wid([y_i/B'(x_i)])$.
  
  We compute $ \wid([y_i/B'(x_i)])$ based on Prop.~\ref{prop:complex-intvl}, 
  and so 
  $
  \lg \wid([1/B'(x_i)]) 
  \leq -\lambda_2 - 4\lg\prod_j\Delta_j(x) + 2(2n-1)\tau_2 + 2n - 8 + 3 = -\lambda_3
  $.
  Finally 
  {\scriptsize
  \[
  \begin{array}{lll}
  \wid([y_i/B'(x_i)]) 
  & \leq 2^{\tau_2} \, \wid([1/B'(x_i)])  + 2^{-\lg\prod_j\Delta_j(x)} \, \wid([y_i]) \\
  & \leq 2^{-\lambda_3 + \tau_2 -\lg\prod_j\Delta_j(x)} 
  \leq 2^{-\lambda_4}
  \end{array}
  \]
  }

\item Compute the sum of the rational functions ie 
  $\frac{A_0(X)}{A_1(X)} = \sum_i\frac{A_{i,0}(X)}{A_{i,1}(X)}$.
  
  Using Lemma~\ref{lem:a-sum-rat-func}
  we get
  $\lg\normi{A_0} \leq \tau_2 - \lg\prod_{j}\Delta_j(x) + (n-1)\tau_1 + 4n - \lg{n} - 4$
  and
  $\lg\normi{A_1} \leq n\tau_1 + 4n - \lg{n} - 4$.
  
  For the approximation we have that 
  $\lg\normi{A_{0} - \wt{A}_{0}} 
  \leq -\lambda_4 +\tau_2  -\lg\prod_{j}\Delta_j(x) + (4n-4)\tau_1 + 32n$.
  If we substitute the various values for $\lambda_i$ we have
  $\lg\normi{A_{0} - \wt{A}_{0}} 
  \leq -\lambda + 68 n (\tau_1+3) \lg{n} + 4 n \tau_2 - 6\lg\prod_j\Delta_j(x) + 50 n + 60\lg^{3}{n} + 20$.
\end{enumerate}

The numerator, $A_0$, is the required polynomial $A(X)$.
To achieve an $\ell$-approximation of $A(x)$ we assume 
that we perform all the computations using the maximum precision,
that is $\ell + 68 n (\tau_1+3) \lg{n} + 4 n \tau_2 - 6\lg\prod_j\Delta_j(x) + 50 n + 60\lg^{3}{n} + 20$
or $\lambda = \ell + \OO(n\tau_1\lg{n} + n\tau_2 -\lg\prod_j\Delta_j(x) + \lg^{3}{n})$.

The  overall complexity is
$\OB( n \, \lg^2{n} \, \mu(\lambda) )$
or $\sOB(  n^2\tau_1 + n^2\tau_2 - n \lg\prod_j\Delta_j(x))) $.
\end{proof}

\begin{remark}[The hidden costs]
  In the previous lemma we have assumed bounds on the minimum
  distance between the $x_i$'s, which we denote by $\Delta_i(x)$. 
  The
  complexity results depend on this quantity, 
  as it is very important in the computation  of the number of
  bits that we need to certify the result to a desired accuracy.
  It is reasonable to assume that such bounds are part of the input. 
  
  However, how do we handle the case where such bounds are not known?
  As the precision required
  for the computations depends on these bounds
  we should be able to compute them, given the points $x_i$.
  
  We consider the numbers $x_i \in \CC$ as points on $\RR^2$ and we
  compute their Voronoi diagram.  This costs $\OO(n \lg{n})$
  operations, e.g.~\cite{PrSh-book-85}. 
  Then for each point $x_i$ we find its closest in
  $\OO(\lg{n})$ operations.  Therefore, we can compute the quantities
  $\Delta_i(x)$ in $\OO(n \,\lg{n})$ operations.  But what about the
  required precision? What are the 
  required primitive operations for these
  computations?
  We only need to evaluate the signs of $3 \times 3$ determinants. 
  The precision of Lemma~\ref{lem:lagrange} is 
  sufficient for these operations.
\end{remark}

%!!!! \newpage

\section{Solution of a Cauchy linear system of equations}
\label{sec:Cauchy-solve}

\subsection{Multiplication of a Cauchy matrix by a vector}
\label{sec:Cauchy-mul-vector}

We consider the problem of computing the matrix vector product
 $C {\bf v}$,
where $C = C(s, t) = (\frac{1}{s_i - t_j})_{i,j}^{n-1}$ is a Cauchy matrix 
and $
%y v
{\bf v}= (v_i)_{i=0}^{n-1}$.

We refer the reader to \cite[Problem~3.6.1]{Pan01} for further details
of the algorithm.

Let $\abs{s_i} \leq 2^{\tau_1}$, $\abs{t_i} \leq 2^{\tau_2}$,
$\abs{v_i} \leq 2^{\tau_3}$,
and $\Delta_i(t) = \min_j\abs{t_i - t_j}$,
for all $i$.

The following quantities are also useful
$\Delta_j(s, t) = \min_i\abs{s_j - t_i}$ and 
$\Delta(s, t) = \min_j\Delta_j(s, t)$.

\begin{lemma}
  \label{lem:Cauchy-vec-I}
  If the input is given as a $\lambda$-approximation, where 
  $
  \lambda = \ell + 90n(\tau_1 + 3)\lg{n} + 32(n-1)\tau_2 \lg{n} + 30\tau_3\lg{n} -35
  -24\lg\Delta(s, t) -4\lg\prod_k\Delta_k(t),
  $
  then we can compute an
  $\ell$-approximation of the vector 
  $C{\bf v}$
  in
  $\OB( n \, \lg^{2}{n} \, \mu(\lambda))$
  or 
  $\sOB( n^2\tau_1 + n^2\tau_2 + n\tau_3 - n\lg\Delta(s, t) - n\lg\prod_k\Delta_k(t))$.

  % ET: no need for the moment {\bf TODO:} Give an upper bound for the output.
\end{lemma}
\begin{proof}
The input is given as 
$\lambda$-approximations, where $\lambda$ is
to be specified in the sequel.
We track the loss of accuracy at each step of the algorithm.

\begin{enumerate}[leftmargin=5pt,rightmargin=3pt]
\item Consider the rational functions $\frac{v_k}{x-t_k} = \frac{P_k}{Q_k}$.

 For the numerators it holds $\dg{P_k} = 0$, $\normi{P_k} \leq 2^{\tau_3}$,
 and $\lg\normi{P_k - \wt{P}_k} \leq - \lambda$. 
 For the denominators it holds $\dg{Q_k} = 1$, $\normi{Q_k} \leq 2^{\tau_2}$,
 and $\lg\normi{Q_k - \wt{P}_k} \leq - \lambda$. 
 
\item Compute the sum $\frac{P}{Q} = \sum_k\frac{P_k}{Q_k}$.
  For this computation we rely on Lemma~\ref{lem:a-sum-rat-func}.
  
  For the numerator of the result we have $\dg{P} \leq n-1$, 
  $\lg\normi{P} \leq \tau_3 + (n-1)\tau_2 + 5n - \lg{n} - 4$,
  and
  $\lg\normi{P - \wt{P}} \leq -\lambda + \tau_3 + (4n-4)\tau_2 + 32n = -\lambda_1$.

  For the denominator of the result we have $\dg{Q} \leq n$, 
  $\lg\normi{Q} \leq n\tau_2 + 4n - \lg{n} - 4$,
  and
  $\lg\normi{Q - \wt{Q}} \leq -\lambda + \tau_3 + (4n-4)\tau_2 + 32n = -\lambda_1$.

\item Compute $P(s_i)$ and $Q(s_i)$ for all $i$.

  For this multipoint evaluation we use Lemma~\ref{lem:rem-m-polys}
  (with $m = n$, $n=1$, $\tau_1 = \lg\normi{P}$, $\rho = \tau_1$).

  We have
  $\lg\abs{P(s_i)} \leq (\tau_1 + 1)n + (n-1)\tau_2 + \tau_3 + 4n - 4$
  and
  $\lg\abs{P(s_i) - \wt{P}(\wt{s}_i)} \leq 
  -\lambda_1 + (\tau_3 + (n-1)\tau_2 + 4n - \lg{n} - 4)\lg{n} + 60n(\tau_1 + 3)\lg{n} + 60\lg^{3}{n}
  = -\lambda_2$.

  Similar bounds hold for $Q(s_i)$.

\item Compute the fractions $\frac{P(s_i)}{Q(s_i)}$.
  These are the elements of the result of the matrix-vector multiplication 
  $C{\bf v}$.
  
  For this task we need to perform $n$ (complex) divisions.  We use
  complex interval arithmetic to compute the 
  loss
  of precision,
  as we did for deriving the bounds for Lagrange interpolation.
  To compute a lower bound for $\abs{Q(s_i)}$
  we use Lemma~\ref{lem:poly-eval-lower-bound}, and so 
  \[
  \begin{aligned}
    \abs{Q(s_i)} & \geq \Paren{\Delta_i(s, t)}^6 \, 2^{-6 \, \lg{\normi{Q}} - 6\lg{n}} \, 2^{\lg\prod_k\Delta_k(t) - 6}  \\
    & \geq \Paren{\Delta_i(s, t)}^6 \, 2^{\lg\prod_k\Delta_k(t)}\, 2^{-6n\tau_2 - 24n +20} \\
    &\geq 2^{-\nu}
  \end{aligned}
  \]

  Let $\abs{Q(s_i)} \leq 2^{T}$, where $T = n\tau_1 + n\tau+2 +4n -4$.
  Using Prop.~\ref{prop:complex-intvl}, 
  $\wid{[1/Q(s_i)]} \leq 2^{4\nu + 2T + 3} \, 2^{-\lambda_2} \leq 2^{-\lambda_3}$
  and
  $
  \wid{[P(s_i)/Q(s_i)]} \leq 
  2 \, \abs{P(s_i)} \, 2^{-\lambda_3} + 2\, \abs{1/Q(s_i)} \, 2^{-\lambda_2}
  \leq 2^{-\lambda_3 + T + 2}
  $
  which is also the accuracy of the result.
\end{enumerate}

Putting together the various values of $\lambda_i$, $\nu$, and $T$, 
we achieve an $\ell$ approximation
by choosing 
$\lambda = \ell + 90n(\tau_1 + 3)\lg{n} + 32(n-1)\tau_2 \lg{n} + 30\tau_3\lg{n} -35
-24\lg\Delta(s, t) -4\lg\prod_k\Delta_k(t)$.

We perform all the computations
 with maximum required accuracy,
and so the overall complexity is 
$\OB( n \, \lg^{2}{n} \, \mu(\lambda))$
or 
$\sOB( n^2\tau_1 + n^2\tau_2 + n\tau_3 - n\lg\Delta(s, t) - n\lg\prod_k\Delta_k(t))$.
\end{proof}

%!!!
%\begin{remark}[The hidden costs]
%!!!  In the previous theorem we have assumed that 
%!!! \end{remark}

\subsection{Trummer's problem}
\label{sec:Trummer}

This is the important special case where 
$s = t$ and the diagonal entries of the Cauchy matrix are set to zero.
In this case we compute the matrix vector product 
by using the following formula,
\begin{equation}
  \label{eq:trummer}
  (C \, v)_{i=0}^{n-1} = \Paren{ \frac{2\, P'(s_i) - v_i \, Q^{''}(s_i)}{2\, Q'(s_i)} }_{i=0}^{n-1}
  = \Paren{\frac{A_{0,i}}{A_{1,i}}}_{i=0}^{n-1}
\end{equation}
We refer the reader to \cite[Problem~3.6.3]{Pan01} for further details.

\begin{corollary}
  \label{cor:trummer}
  Using the notation of Lemma~\ref{lem:Cauchy-vec-I}, 
  we can solve Trummer's problem in $\OB(n \, \lg^{2}{n} \, \mu(\lambda))$,
  where
  $\lambda = \ell + 70(\tau_1 + 3)n\lg{n} +4\tau_3 \lg{n} - 4\lg\prod_j\Delta(s)$.
\end{corollary}

\begin{proof}
  First we compute bounds for the numerator of Eq.~(\ref{eq:trummer}).
  Following the proof of Lemma~\ref{lem:Cauchy-vec-I} we have
  $\lg\abs{P'(s_i)} \leq (\tau_1 + 1)n + (n-1)\tau_2 + \tau_3 + 4n - 4 + \lg{n}$
  and
  $\lg\abs{P(s_i) - \wt{P}(\wt{s}_i)} \leq -\lambda_2 + \log{n}$.
  Similarly for 
  $\lg\abs{Q''(s_i)} \leq (\tau_1 + 1)n + (n-1)\tau_2 + \tau_3 + 4n - 4 + 2\lg{n}$
  and
  $\lg\abs{Q''(s_i) - \wt{Q''}(\wt{s}_i)} \leq -\lambda_2 + 2\log{n}$.
  For the first derivative we add to the logarithm of the norm a term $\lg{n}$
  and for the second a term $2\lg{n}$. 
  Moreover, $\tau_1 = \tau_2$, as $s= t$.

  Taking into account that $\abs{v_i} \leq 2^{\tau_3}$ we deduce that
  \[
  \lg\abs{A_{0,i}} \leq (\tau_1 + 1)n + (n-1)\tau_2 + 2\tau_3 + 4n - 4 + 2\lg{n}
  \]
  and 
  $\lg\abs{A_{0,i} - \wt{A}_{0,i}} \leq -\lambda_2 + \tau_3 + 2\log{n} + 2$~.
  
  Regarding the denominator we have that 
  $\lg\abs{Q'(s_i)} \leq (\tau_1 + 1)n + (n-1)\tau_2 + \tau_3 + 4n - 4 + \lg{n} = T$.
  In addition
  $\abs{Q'(s_i)} = \prod_{j \not= i}\abs{s_i - s_j} \geq \prod_{j \not= i}\Delta_j(s)$  
  and so 
  $\abs{1/ A_{1,i}} = \abs{1/2 Q'(s_i)} \leq \prod_{j \not= i}(\Delta_j(s))^{-1}$.
  Prop.~\ref{prop:complex-intvl} leads to 
  $
  \lg\wid([1/A_{1,i}]) = \lg \wid{[1/Q'(s_i)]} \leq -\lambda_2 + \lg{n} -4 \lg\prod_{j \not= i}\Delta_j(s) + 2T + 3
  $
  where $\lambda_2$ is defined at the (C3) step of the proof of Lemma~\ref{lem:Cauchy-vec-I}.
  
  Putting all the pieces together, we have 
  \[
  \begin{array}{ll}
  \wid([ \frac{A_{0,i}}{A_{1,i}} ])  \leq &
  2\,\abs{A_{0,i}} \, 2^{-\lambda_2 + \lg{n} -4 \prod_{j \not= i}\Delta_j(s) + 2T + 3} \\
  & + 2\,\abs{1/A_{1,i}} \, 2^{-\lambda_2 + \tau_3 + 2\log{n} + 2} 
  \end{array}
  \]
  and after many simplifications and overestimations
  \[
  \begin{array}{ll}
  \lg\Abs{\frac{A_{0,i}}{A_{1,i}} -\frac{\wt{A}_{0,i}}{\wt{A}_{1,i}}} & 
  \leq \wid([ \frac{A_{0,i}}{A_{1,i}} ]) \\
  & \leq
  -\lambda + 70(\tau_1 + 3)n\lg{n} +4\tau_3 \lg{n} - 4\lg\prod_j\Delta(s).
  \end{array}
  \]

  To achieve an $\ell$-approximation we need the input to be a $\lambda$-approximation,
  where $\lambda = \ell + 70(\tau_1 + 3)n\lg{n} +4\tau_3 \lg{n} - 4\lg\prod_j\Delta(s)$,
  and we perform all the computations with this number of bits. 
  The overall complexity is $\OB(n \, \lg^{2}{n} \, \mu(\lambda))$.
\end{proof}

\subsection{Solving a (Cauchy) linear system}

We consider the following problem.
\begin{problem}[Cauchy linear system of equations]
  \label{prob:Cauchy-linear-system}
  Solve a non-singular Cauchy linear system of $n$ equations, $C({\bf
    s,t}) \, {\bf v}={\bf r}$ for an unknown vector $\bf v$ and 3
  given vectors $\bf r,s$, and $\bf t$.
\end{problem}

\begin{theorem}
  \label{thm:Cauchy-linear-system}
  Let $\abs{s_i} \leq 2^{\tau_1}$, $\abs{t_i} \leq 2^{\tau_2}$,
  $\abs{r_i} \leq 2^{\tau_3}$,
  and $\Delta_i({\bf s}) = \min_j\abs{s_i - s_j}$,
  $\Delta_i({\bf t}) = \min_j\abs{t_i - t_j}$,
  for all $i$.
  Let also $\Delta({\bf s}, {\bf t}) = \min_{i,j}\abs{s_i - t_j}$.
  
  If the input is given $\lambda$-approximations,
  where $\lambda = \ell + 630(\tau_1 + \tau_2)n\lg{n} + 32\tau_3\lg{n} -35 
  -35\lg{n}\lg\prod\Delta_j({\bf s}) -5\lg\prod\Delta_j({\bf t})
  -25\lg\Delta({\bf s}, {\bf t}) $,
  then an $\ell$-approximation solution to
  Problem~\ref{prob:Cauchy-linear-system} could be obtained in 
  $\OB( n \, \lg^2{n} \, \mu(\lambda))$,
  or $\sOB( n \ell + n^2(\tau_1 + \tau_2) + n\tau_3
  -n\lg\prod\Delta_j({\bf s}) -n\lg\prod\Delta_j({\bf t})
  -n\lg\Delta({\bf s}, {\bf t}))$.  
\end{theorem}

\begin{proof}
Following \cite[Eq.~3.6.10]{Pan01} then inverse of 
$C({\bf s,t})$ could be obtained as follows
\begin{equation}
  \label{eq:Cauchy-inv}
  \begin{array}{ll}
  C^{-1}({\bf s,t}) & =\diag(p_{\bf s}(t_i)/p'_{\bf t}(t_i))_{i=0}^{n-1}
  \, C({\bf t,s}) \, 
  \diag(p_{\bf t}(s_i)/p'_{\bf s}(s_i))_{i=0}^{n-1} \\
  & = D_1 \cdot C({\bf t,s}) \cdot D_2
  \end{array}
\enspace ,
\end{equation}
where 
$p_{\bf t}(X) = \prod_i(X - t_i)$ and 
$p_{\bf s}(X) = \prod_i(X - s_i)$.

Using this
formula we can solve the linear system as
\[
{\bf v} = C^{-1}({\bf s,t}) \, {\bf r} =
D_1 \cdot C({\bf t,s}) \cdot D_2 \, {\bf r} 
\]

We apply Lemma 17 to analyze 
 the computation of the polynomials $p_s(x)$ and  $p_t(x)$.
Then we compute the two diagonal matrices $D_1$ and $D_2$ by
applying the Moenck--Borodin algorithm for multipoint evaluation
(see Section 4). 
At first we perform $n$ ops to compute the vector 
%!!! 
%The computation of the matrix-vector product $D_2 \, {\bf r}$ 
%demands $n$ operations and results 
%in a vector 
${\bf r_1}=D_2 \, {\bf r}$.
Then we
%!!!  we perform the Cauchy matrix-vector product, 
use Lemma~\ref{lem:Cauchy-vec-I} to
obtain
 the vector ${\bf r_2}=C({\bf t,s}){\bf r_1}$. Finally we 
%!!! perfrom the 
multiply the matrix $D_1 \,$ by  ${\bf r_2}$ 
to obtain the vector ${\bf v}$.
We track the loss of precision at each of these 
%!!!
three steps.

\begin{enumerate}[leftmargin=5pt,rightmargin=3pt]
\item  Computation of the matrices $D_1$ and $D_2$.

For this task we need to compute
$ p_{\bf s}(t_i)$, $p'_{\bf s}(s_i)$,
$p_{\bf t}(s_i)$, and $p'_{\bf t}(t_i)$.

It holds 
$\lg\normi{p_{\bf s}} \leq n\tau_1 + 4n - \lg{n}-4$
and $\lg\normi{p_{\bf s} - \wt{p}_{\bf s}} \leq -\lambda+(4n-4)\tau_1 +32n$, 
using Lemma~\ref{lem:a-mul-m-polys}.
By applying Lemma~\ref{lem:rem-m-polys} we get 
$\lg\normi{p_{\bf s}(t_i)} \leq n\tau_1 + n\tau_2 + 5n - 3$.
and $\lg\normi{p_{\bf s}(t_i) - \wt{p}_{\bf s}(\wt{t}_i)} \leq 
-\lambda+ ( n\lg{n} + 4\,n-4 )\tau_1 + 60\, n \tau_2\lg{n} -4\,\lg{n} +184\,n\lg{n}+32\,n- ( \lg{n}  ) ^{2}$.

However, to simplify the calculations we consider the inferior bounds
$\lg\normi{p_{\bf s}(t_i)} \leq 7n(\tau_1 +\tau_2)$.
and $\lg\normi{p_{\bf s}(t_i) - \wt{p}_{\bf s}(\wt{t}_i)} \leq 
-\lambda + 300 (\tau_1 + \tau_2)n \lg{n}$
for all the involved quantities.

We also need lower bounds for
$\abs{p'_{\bf t}(t_i)}$ and $\abs{p'_{\bf s}(s_i)}$.

It holds 
$\abs{p'_{\bf t}(t_i)} \geq \prod_{j\not=i}\abs{t_i - t_j} \geq  \prod_{j\not=i}\Delta_j({\bf t})$.
Similarly
$\abs{p'_{\bf s}(s_i)} \geq \prod_{j\not=i}\abs{s_i - s_j} \geq  \prod_{j\not=i}\Delta_j({\bf s})$.

This leads to the bounds:
$\Abs{\frac{p_{\bf s}(t_i)}{p'_{\bf t}(t_i)}} \leq 2^{7n(\tau_1 + \tau_2) - \lg\prod_{j\not=i}\Delta_j({\bf t})}$
and 
$\Abs{\frac{p_{\bf t}(s_i)}{p'_{\bf s}(s_i)}} \leq 2^{7n(\tau_1 + \tau_2) - \lg\prod_{j\not=i}\Delta_j({\bf s})}
$.

By combining the previous bounds with the complex interval arithmetic
of Prop.~\ref{prop:sep-bounds} we obtain the following estimation 
for the approximation:
{\scriptsize
\[
\lg\Abs{\frac{p_{\bf t}(s_i)}{p'_{\bf s}(s_i)} - {\widetilde{\frac{p_{\bf t}(s_i)}{p'_{\bf s}(s_i)}}}} 
\leq -\lambda + 315 (\tau_1 + \tau_2)n \lg{n} - 4\lg\prod_{j\not=i}\Delta_j({\bf s})
\]
\[
\lg\Abs{\frac{p_{\bf s}(t_i)}{p'_{\bf t}(t_i)} - {\widetilde{\frac{p_{\bf s}(t_i)}{p'_{\bf t}(t_i)}}}} 
\leq -\lambda + 315 (\tau_1 + \tau_2)n \lg{n} - 4\lg\prod_{j\not=i}\Delta_j({\bf t})
\]
}
\item ${\bf r_1} = D_2 \, {\bf r}$ .
  
  This computation increases the bounds by a factor of $\tau_3$.
  To be more specific, for the elements of ${\bf r_1}$, $r_{1,i}$ we have
  $ \abs{r_{1,i}} \leq 2^{7n(\tau_1 + \tau_2) + \tau_3 - \lg\prod_{j\not=i}\Delta_j({\bf s})} $
  and 
  \[ \lg\abs{r_{1,i} - \wt{r}_{i,1}} \leq 
  -\lambda + 316 (\tau_1 + \tau_2)n \lg{n} + \tau_3 - 4\lg\prod_{j\not=i}\Delta_j({\bf s})\]
\item ${\bf r_2} =  C({\bf t,s}) \, {\bf r_1}$. \\
  For this computation we need to apply Lemma~\ref{lem:Cauchy-vec-I}.
  We obtain
  {\scriptsize
  \[ \lg\abs{r_{2,i}} \leq 7n(\tau_1 + \tau_2) + \tau_3
   -\lg\prod\Delta_j({\bf s}) - \lg\Delta({\bf s}, {\bf t}), \]
  \[ 
  \begin{aligned}
  \lg\abs{r_{2,i} - \wt{r}_{2,i}} \leq & 
  -\lambda + 616(\tau_1 + \tau_2)n\lg{n} + 31\tau_3\lg{n} -35 \\
  & - 34\lg{n}\lg\prod\Delta_j({\bf s})
  - 4 \lg\prod\Delta_j({\bf t})
  - 24  \lg\Delta({\bf s}, {\bf t}) 
  \end{aligned}
  \]
  }
\item ${\bf v} = D_1 \, {\bf r_2}$. \\

  This computations leads to the following bounds
  {\scriptsize
  \[ 
  \lg\abs{v_i} \leq 14n(\tau_1 + \tau_2) + \tau_3
  -\lg\prod\Delta_j({\bf s}) -\lg\prod\Delta_j({\bf t}) - \lg\Delta({\bf s}, {\bf t}) 
  \]
   \[
   \begin{aligned}
     \lg\abs{v_{i} - \wt{v}_{i}} \leq &
     -\lambda + 630(\tau_1 + \tau_2)n\lg{n} + 32\tau_3\lg{n} -35 \\
     &\quad\quad - 35\lg{n}\lg\prod\Delta_j({\bf s})
     - 5\lg\prod\Delta_j({\bf t})
     - 25 \lg\Delta({\bf s}, {\bf t}) 
   \end{aligned}
  \]
  }
\end{enumerate}

To achieve an $\ell$-approximation of the output we should
require a $\lambda$-approximation of the input, 
where $\lambda =  \ell + 630(\tau_1 + \tau_2)n\lg{n} + 32\tau_3\lg{n} -35 
-35\lg{n}\lg\prod\Delta_j({\bf s}) -5\lg\prod\Delta_j({\bf t})
-25\lg\Delta({\bf s}, {\bf t}) $.

As in all the previous sections we perform all the computations using 
the maximum precision.
The overall complexity of the algorithm
is $\OB( n \, \lg^2{n} \, \mu(\lambda))$,
or $\sOB( n \ell + n^2(\tau_1 + \tau_2) + n\tau_3
-n\lg\prod\Delta_j({\bf s}) -n\lg\prod\Delta_j({\bf t})
-n\lg\Delta({\bf s}, {\bf t}))$.
\end{proof}

{\scriptsize
  \textbf{Acknowledgments.}
  VP is supported by NSF Grant CCF--1116736 and
  PSC CUNY Awards 64512--0042 and 65792--0043.
  ET is partially supported by 
  GeoLMI 
  (ANR 2011 BS03 011 06), 
  HPAC (ANR ANR-11-BS02-013)
  and an
  FP7 Marie Curie Career Integration Grant.  
}

{
  \scriptsize
   \bibliographystyle{abbrv}  
   \bibliography{refine} 
   %\bibliography{et,Personal,algebra,complexity,emiris,arithm,numeric,geometry,soft,ecg}
}

\newpage
%\newpage
\section*{Appendix}

{\it 
  \textbf{Lemma~\ref{lem:T-inv}.}
  Let $n+1=2^k$
  for a positive integer $k$
  and let $T$ be a lower 
  triangular Toeplitz $(n+1)\times (n+1)$ matrix of. Eq.~(\ref{eq:T-lin-eq}),
  having ones on the diagonal. Let its subdiagonal entries be
  complex numbers
  of magnitude at most $2^{\tau}$ known up to a precision
  $2^{-\lambda}$. 
  Let $2^{\rho}$ be an upper bounds on the magnitude of the roots of the 
  univariate polynomial 
  $t(x)$ associated with $T$.
  Write $T^{-1}=(T^{-1}_{i,j})_{i,j=0}^{n}$.
  Then 
  \[
  \max_{i,j} \abs{T^{-1}_{i,j}} \leq   2^{(\rho+1)n + \lg(n)+1} 
  \enspace .
  \]

  Furthermore, to compute the
  entries of $T^{-1}$ up to the
  precision of $\ell$ bits, 
  that is to compute a matrix $\wt{T}^{-1}=(\wt{T}^{-1}_{i,j})_{i,j=0}^{n}$ 
  such that 
  %\begin{displaymath}
  $  \max_{i,j} \abs{T^{-1}_{i,j} - \wt{T}^{-1}_{i,j}} \leq 2^{-\ell} $,
  %  \enspace ,
  %\end{displaymath}
  it is sufficient to know the entries of $T$ up to
  the precision of

  $\ell + 10\tau \lg{n} + 70\lg^2{n} + 8(\rho+1)n\lg{n}$
  or $\OO(\ell + (\tau + \lg{n} +n \rho)\lg{n}) = \sO(\ell + \tau + n\rho)$ bits.

  The computation of $\wt{T}^{-1}$ costs
  $\OB( n \, \lg^{2}(n) \, \mu(\ell + (\tau + \lg{n} +n \rho)\lg{n}))$ 
  or $\sOB(n \ell + n\tau + n^2 \rho)$.
}

\begin{proof}[of Lemma~ \ref{lem:T-inv}]
  We will prove the claimed estimates by 
reducing the
  inversion to recursive multiplication of polynomials defined by
  equation (\ref{eq:T-inv}) and Lemma \ref{lem:veceq}.
  
   Consider the $\frac{n+1}{2}\times \frac{n+1}{2}$ 
   Toeplitz matrices $T_0^{-1}$ 
   and $T_1$ of equation (\ref{eq:T-inv}). 
   The Toeplitz  matrix $T_0^{-1}$ is triangular, and so
its first column, ${\bf p}=(p_{i})_{i=0}^{(n-1)/2}$, with $p_0=1$, 
defines this matrix and the polynomial $p(x)=\sum_{i=0}^{(n-1)/2}p_ix^i$ of degree $(n-1)/2$.
 Likewise the vector ${\bf t}=(t_{n-i})_{i=1}^{n}$ 
(made up of two overlapping vectors, that is,
the reversed first row, $(t_{n-1},t_{n-2},\dots,t_{(n-1)/2})$
of the matrix $T_1$
 and its first column,   $(t_{(n-1)/2},t_{(n-3)/2},\dots,t_{0})^T$)
defines this matrix and the
 polynomial $\tilde t(x)=\sum_{i=1}^{n}t_{n-i}x^{i-1}$ of degree $n-1$. 
 The first column of the $\frac{n+1}{2}\times \frac{n+1}{2}$ 
 Toeplitz matrix $T_1T_0^{-1}$ is a subvector, ${\bf v}$, 
of dimension $(n+1)/2$ of the coefficient vector of the polynomial 
product $\tilde t(x)p(x)$, having degree $3(n-1)/2$. Likewise the first 
column of the $\frac{n+1}{2}\times \frac{n+1}{2}$
matrix 
$-T_0^{-1}T_1T_0^{-1}$
is the vector ${\bf q}=-T_0^{-1}{\bf v}$, which is a subvector of
dimension $(n+1)/2$ of the coefficient vector of the polynomial
product $p(x)v(x)$ of degree $n-1$, where the polynomial $v(x)$
of degree $(n-1)/2$ is defined by its coefficient vector ${\bf v}$. In
sum the vector ${\bf q}$ is the coefficient vector of a polynomial
$\tilde q(x)$ obtained by two successive multiplications of
polynomials, 
each followed by the truncation of the coefficient vectors. 
Namely we first compute the polynomial $\tilde t(x)p(x)$, then
truncate it to obtain the polynomial $v(x)$, then compute the
polynomial $-p(x)v(x)$, and finally truncate it to obtain the
polynomial $\tilde q(x)$.
 
The truncation can only decrease the degree of a polynomial and the maximum length of 
its coefficients, and so
we can bound the precision and the cost of computing the matrix product $-T_0^{-1}T_1T_0^{-1}$
by the bounds on
 the precision and the cost of computing the polynomial product $P_0^2 \,P_1$,
estimated in Corollary \ref{cor:p02p1-mul} 
for $P_0=p(x)$, $P_1=\tilde t(x)$, and $d=(n-1)/2$.

  First we prove the upper bound on the elements of $T^{-1}$.
  Consider the division 
  $\phi_{i,k}(x) = t(x) q(x) + r(x)$,
  where $t(x)$ is the univariate polynomial of degree $n$ 
associated with
%!!! that corresponds to $T$,
  $\phi_{i,k}(x) = \sign(T_{i,1}^{-1}) x^n + \sign(T_{i,k}^{-1}) \, x^k$,
  $2 \leq k \leq n$ and $1 \leq i \leq n$.
  Then $\norm{\phi_{i,k}}_{\infty} \leq 1$
  and $\norm{\phi_{i,k}}_{2} \leq \sqrt 2$.
  By abusing notation we also write  $\phi_{i,k}$ and $q$ to denote the 
  coefficient vectors
  of these polynomials.
  Using Eq.~(\ref{eq:T-lin-eq}) 
  we can compute the elements of $q$ 
  from the equation $q = T^{-1} \phi_{i,k}$.
  In this way
  $\abs{q_i} = \abs{  \sign(T_{i,1}^{-1}) T_{i,1}^{-1} +  \sign(T_{i,k}^{-1}) T_{i,k}^{-1}} 
  = \abs{ T_{i,1}^{-1}} +  \abs{T_{i,k}^{-1}} $.

  Using Lemma~\ref{lem:quo-rem-bd} we 
  obtain
  the inequality
  $\abs{q_i} \leq \norm{q}_{\infty} \leq 2^{n + \lg{n} + n\rho} \norm{\phi_{i,k}}_{\infty}$,
  which in turn implies
  \begin{equation}
    \label{eq:T-inv-up-bd}
    \abs{ T_{i,k}^{-1}} \leq 
    \abs{ T_{i,1}^{-1}} +  \abs{T_{i,k}^{-1}} = 
    \abs{q_i} \leq 2^{n + \lg{n} + n\rho} \,\norm{\phi_{\nu,k}}_\infty 
    \leq  2^{n + \lg{n} + n\rho}
    \enspace.
  \end{equation}

  Let $P_0$ and $P_1$ 
  denote  the two polynomials defined earlier in the proof
  such that $\max_{i,j}\{|T_{i,j}^{-1}|\}\le ||P_0^2P_1||_{\infty}$.

Then   Eq.~(\ref{eq:T-inv-up-bd}) implies 
the inequality
  \begin{equation}
    \label{eq:po2p1-lg-bd}
    \lg{ \Norm{ P_0^2\, P_1}_{\infty} }  \leq
    {n + \lg{n} + n\rho} = N
    \enspace.
  \end{equation}
  
  Recall that $n+1=2^k$ 
  by assumption and that the polynomial $P_0^2 \,P_1$
  has degree $2n-2$.
  Write 
  $h=\lceil \lg (2n-2) \rceil$.
  Under the assumption that the input elements are known up to 
  the precision of $\lambda$ bits,
  we compute a polynomial ${\widetilde {P_0^2\, P_1}}$
  such that 
  {\scriptsize
  \begin{equation}
    \label{eq:po2p1-lg-prec}
    \begin{array}{l}
      \lg{ \Norm{P_0^2 P_1 - {\widetilde {P_0^2\, P_1}}}_{\infty} }  \leq \\
       \leq -\lambda + (2\lg{n} + 8)\tau + 2\lg{n}(4\lg{n}+27) + 8(\lg{n}-1) N     \\
       \leq -\lambda + (2k + 8)\tau  + 2k(4k+27)  + 8(k-1)N
    \enspace .
    \end{array}
  \end{equation}  
}

We proceed by induction.
Let the base case be $n+1 = 4$; then $k = 2$.
We need to invert the following matrix
{\scriptsize
\[
T = \left[ 
  \begin {array}{cccc} 
    1&0&0&0\\ \noalign{\medskip}t_{{2}}&1&0&0\\ \noalign{\medskip}
    t_{{1}}&t_{{2}}&1&0\\ \noalign{\medskip}
    t_{{0}}&t_{{1}}&t_{{2}}&1
  \end {array} 
\right]  
=
\left[
  \begin{array}{cc}
    T_0 & 0 \\
    T_1 & T_0
  \end{array}
\right] \enspace, 
\text{ where } 
\]
\[
T_0 = \left[ 
  \begin {array}{cc} 1&0\\ \noalign{\medskip}t_{{2}}&1 \end {array} 
\right]
\enspace ,
T_0^{-1} = \left[ 
  \begin {array}{cc} 1&0\\ \noalign{\medskip} -t_{{2}}&1\end {array}
\right]
\enspace ,
T_1 = \left[ 
  \begin {array}{cc} t_1& t_2\\ \noalign{\medskip} 
t_0 & t_1\end {array}
\right]
\enspace ,
\]
}
and $\abs{t_i} \leq 2^{\tau}$.
The associated polynomials are 
$P_0(x) = 1 - t_2 x$, for  $T_0^{-1}$,
and  $P_1(x) = t_2 + t_1 x + t_0 x^2$,
for  $T_1$. Therefore 
$P_1P_0=(1-t_2x)(t_2+t_1x+t_0x^2)=t_2+(t_1-t_2^2)x+(t_0-t_1t_2)x^2-t_0t_2x^3$.
The subvector ${\bf v}=(t_1-t_2^2,t_0-t_1t_2)^T$ of the coefficient vector 
of the polynomial product $P_1P_0$
is the first column of the matrix product $T_1T_0^{-1}$. 
Furthermore the vector 
$-T_0^{-1}{\bf v}=(t_2^2-t_1, 2t_1t_2-t_2^3-t_0)^T$
is a subvector of the coefficient vector of the polynomial product
$-P_0v=-(1-t_2x)(t_1-t_2^2+(t_0-t_1t_2)x)$
where $v$ denotes the polynomial $t_1-t_2^2+(t_0-t_1t_2)x$ with the 
coefficient vector ${\bf v}$.
As we proved earlier, the subvector is the first column of the 
matrix 
\[
-T_0^{-1} \, T_1 \, T_0^{-1} = 
\left[ 
  \begin {array}{cc} 
    {t_{{2}}}^{2}-t_{{1}}&-t_{{2}}\\ \noalign{\medskip}
    -{t_{{2}}}^{3}+2\,t_{{1}}t_{{2}}-t_{{0}}&{t_{{2}}}^{2}-t_{{1}}
  \end {array} 
\right] 
\enspace .
\]

%y
(This is not a subvector
of the
coefficient vector of the polynomial
\[
\begin{array}{l}
P_0^2 P_1=(1-t_2x)^2(t_2+t_1x+t_0x^2) \\
= t_2+(t_1-2t_2^2)x+(t_0-2t_1t_2+t_2^3)x^2+(t_2^2t_1-2t_2t_0)x^3+t_2^2t_0x^4
\enspace
\end{array}
\]
because multiplication of polynomials and the truncation of their
coefficient vectors are not commutative operations,
but as we observed, we still can reduce our 
 study 
of the product $T_0^{-1}T_1T_0^{-1}$ to the study of $p_0^2p_1$.)

We perform the multiplications
by using the algorithm of Lemma~\ref{lem:a-poly-mult}
%!!! ,
% (at first for $A=B=P_0$ and then for $A=P_0^2$ and $B=P_1$)
and the bounds from Cor.~\ref{cor:p02p1-mul},
%!!!
%(for $\tau_0=\tau_1+\tau$, $\nu=\lambda$, and $d=2$)
to obtain a polynomial 
$\widetilde{P_0^2\,P_1}$ 
 such that 
\[
\lg \norm{P_0^2 P_1 - \widetilde{P_0^2 \, P_1}}_{\infty} 
\leq  -\lambda + 10 \tau + 15 \lg{1} + 40
\leq  -\lambda + 12 \tau + 140 +8N
\enspace ,
\]
where the last inequality 
is obtained by substituting $k=2$ in Eq.~(\ref{eq:po2p1-lg-prec}).

It remains to prove the induction.
Assume that the claimed bounds are true for $n+1=2^{k}$
and extend them to $n + 1 = 2^{k + 1}$.
In 
our
case $P_0$ is a polynomial of degree $2^{k-1}-1$.
By induction hypothesis we know the coefficient of $P_0$ within $2^{-\ell}$
for $\ell$ defined by (\ref{eq:po2p1-lg-prec}),
and we can 
apply
Eq.~(\ref{eq:po2p1-lg-bd})
to obtain
$\norm{P_0}_{\infty} \leq 2^{N}$.

The polynomial $P_1$ has degree $2^{k}-2$,
$\norm{P_1}_{\infty} \leq 2^{\tau}$.
Its coefficients are the entries of the matrix $T$ and we know them within $2^{-\lambda}$.
%Since it only involves coefficients of the initial matrix $T$,
%we can assume that we know its coefficients up to the same precision as $P_0$.

We 
apply Cor.~\ref{cor:p02p1-mul} for
 $d = 2^{k-1}-1$, 
$\tau_0 = N$,
$\tau_1 = \tau$, and
$-\nu = -\lambda + (2(k-1) + 8)\tau  + 2(k-1)(4(k-1)+27)  + 8(k-2)N$,
substitute $k\ge 3$, and obtain the following 
approximation bound,
% coefficients of the 
%polynomial $P_0^2 P_1$,
% ET: I will recheck for this
\begin{displaymath}
\begin{array}{ll}
  \lg{ \Norm{P_0^2 P_1 - {\widetilde {P_0^2\, P_1}}}_{\infty} }  
  & \leq -\nu + 8\tau_0 + 2\tau_1 + 14 k  + 42  \\
  & \leq -\lambda +  (2k + 8)\tau  + 8k^2-2k+104+ 8(k-1)N
  \enspace .
\end{array}
\end{displaymath}

%!!!
%  {\bf Elias, this change of 8 into 9 changes nothing, but if it is correct 
%we should include also the implications throughout

%ET: It is 8, there was a typo}.

\medskip

These inequalities imply that all our computations require a precision
bound of at most
$ \lambda'=\ell + (2k + 8)\tau  + 2k(4k+27)  + 8(k-1)N$. 
To simplify the formula, we substitute $k=\lg (n)$ 
rather than $k=\lg (n+1)$ and then rewrite $ \lambda'$ as follows,
$\ell +  (2\lg{n} + 8)\tau  + 2\lg{n}(4\lg{n}+27)  + 8(\lg{n}-1)(n + \lg{n} + n\rho + 1)$
bits, which we simplify
to the bound of $\ell + 10\tau \lg{n} + 70\lg^2{n} + 8(\rho+1)n\lg{n}$
or $\OO(\ell + (\tau + \lg{n} +n \rho)\lg{n})$ bits.

To estimate the overall complexity of approximating the first column of the inverse
matrix  $T^{-1}$, we 
notice that we perform $k$ steps 
overall, for $k =\lg (n+1)$,
but we will 
keep writing $k = \lg (n)$
to simplify the notation.

At each step we perform two multiplications of polynomials of degrees
at most $2^{k}$ with the coefficients having absolute values less than
$2^{\OO(n \tau)}$, and we use the precision of
$\ell + 10\tau \lg{n} + 70\lg^2{n} + 8(\rho+1)n\lg{n}
= \OO(\ell + (\tau + \lg{n} +n \rho)\lg{n})$ bits.
All these bounds 
together imply that 
\begin{displaymath}
\begin{array}{l}
  \sum_{k=1}^{\lg{(n+1)}} k \cdot 2 \cdot \OB( 2^{k} \cdot \lg{2^k} \cdot \mu(\ell + (\tau + \lg{n} +n \rho)\lg{n})) = \\
  = \OB( n \, \lg^{2}(n) \, \mu(\ell + (\tau + \lg{n} +n \rho)\lg{n})) 
  \enspace ,
\end{array}
\end{displaymath}
which concludes the proof.
\end{proof}

\section{A normalization of Kirrinnis' results}

The following results express the bounds of Kirrinnis
\cite{kirrinnis-joc-1998} for polynomials of arbitrary norms; for this
an appropriate scaling is applied. The complexity bounds depends on
polynomial multiplication algorithms that rely on Kronecker
substitution and not on FFT as our results that we presented earlier.
The bounds are asymptotically the same as ours, up to logarithmic factors.

The following is from \cite[Theorem~3.7 and Algorithm~5.1]{kirrinnis-joc-1998}.
\begin{lemma}[Multiplication of polynomials]
  \label{lem:k-a-mul-polys}
  Let $P_j \in \CC[x]$ of degree $n_j$ and $\normi{P_j} \leq 2^{\tau}$
  and $\wt{P}_j$ be an approximation of $P_j$ such that 
  $\normi{P_j - \wt{P}_j} \leq 2^{-\lambda}$,
  with $\lambda = \ell n(\tau+2) + n\lg{n}$.,
  where $1 \leq j \leq m$,
  $n_1 \leq \cdots \leq n_m$, and  $\sum_{n_j} = n$.
  We can compute $\prod_j\wt{P}_j$ such that 
  $\normi{ \prod_j P_j - \prod_j \wt{P}_j } \leq 2^{-\ell}$
  in $\OB( \mu(n \cdot \lg{n} \cdot (\ell + n \tau + \sum_j\lg{n_j})))$
  or
  $\sOB(n(\ell + n \tau))$.
\end{lemma}
\begin{proof}
  Let $p_j = 2^{-\tau -\lg{n_j} + n_j} \, P_j$. Then 
  \[
  \normi{P_j} \leq 2^{\tau} \Rightarrow
  \normo{P_j} \leq 2^{\tau + \lg{n_j}}
  \]
  \[ \Rightarrow
  2^{-\tau - \lg{n_j}+ n_j}\normo{P_j} \leq 2^{n_j} \Rightarrow
  \normo{p_j} \leq 2^{n_j} 
  \enspace .
  \]
  Moreover,
  \[
  \normi{P_j - \wt{P}_j} \leq 2^{-\lambda} \Rightarrow
  \normo{P_j - \wt{P}_j} \leq 2^{-\lambda + \lg{n_j}} 
  \]
  \[\Rightarrow
  \normo{p_j - \wt{p}_j} \leq 2^{-\tau - \lg{n_j}+ n_j} 2^{-\lambda + \lg{n_j}}
  = 2^{-\lambda -\tau + n_j}
  \enspace .
  \]
  
  The algorithm of Kirrinnis \cite[Theorem~3.7 and Algorithm~5.1]{kirrinnis-joc-1998}
  computes and approximation of $\prod_j p_j$
  such that $\Normo{ \prod_j p_j - \prod_j \wt{p}_j } \leq 2^{-s + \tau + 3n}$, 
  which implies
  \[
  \Normo{ \prod_j 2^{-\tau -\lg{n_j} +  n_j} P_j -  \prod_j 2^{-\tau -\lg{n_j} +  n_j} \wt{P}_j } 
  \leq 2^{-\lambda + \tau + 3n}
  \]
  \[\Rightarrow
  2^{-n\tau - \sum_j\lg{n_j} + n} \Normo{ \prod_j P_j - \prod_j \wt{P}_j } \leq 2^{-\lambda +\tau + 3n} 
  \enspace ,
  \]
  and so 
  \[
  \Normo{ \prod_j P_j - \prod_j \wt{P}_j } \leq 2^{-\lambda + n(\tau+2) + n\lg{n}} 
  \enspace .
  \]
  If we want the error to be less than $2^{-\ell}$, then 
  we need to choose initial precision $\lambda \geq \ell n(\tau+2) + n\lg{n}$.
  The total cost is
  $\OB( \mu(n \cdot \lg{n} \cdot (\ell + n \tau + \sum_j\lg{n_j})))$
  or
  $\sOB(n(\ell + n \tau))$.
  
\end{proof}

The algorithm for computing the sum of rational functions
\cite[Theorem~3.8 and Algorithm~5.2]{kirrinnis-joc-1998}
relies on Lemma~\ref{lem:k-a-mul-polys} and the bounds are similar, so we do not elaborate further.

% \begin{lemma}[Sum of rational functions]

% Let $P_j \in \CC[x]$ of degree $n_j$ and $\normi{P_j} \leq 2^{\tau_1}$.
% Let $Q_j \in \CC[x]$ of degree (at most) $n_j-1$ and $\normi{Q_j} \leq 2^{\tau_2}$.

% \end{lemma}
% %%
% \begin{proof}
%   We proceed as in the previous lemma.
  
%   We consider the computation $\frac{Q}{P} = \sum_{i}{\frac{Q_i}{P_i}}$

%   $\normi{P_j - \wt{P}_j} \leq 2^{-\lambda}$ and
%   $\normi{Q_j - \wt{Q}_j} \leq 2^{-\lambda}$

%   Let $q_j = q_j$ and $p_j = 2^{-\tau -\lg{n_j} + n_j} P_j$
%   It holds 
%   \[
%   \normi{P_j} \leq 2^{\tau} \Rightarrow
%   \normo{P_j} \leq 2^{\tau + \lg{n_j}} \Rightarrow
%   2^{-\tau - \lg{n_j}+ n_j}\normo{P_j} \leq 2^{n_j} \Rightarrow
%   \normo{p_j} \leq 2^{n_j}
%   \]
%   Moreover,
%   \[
%   \normi{P_j - \wt{P}_j} \leq 2^{-\lambda} \Rightarrow
%   \normo{P_j - \wt{P}_j} \leq 2^{-\lambda + \lg{n_j}} \Rightarrow
%   \normo{p_j - \wt{p}_j} \leq 2^{-\tau - \lg{n_j}+ n_j} 2^{-\lambda + \lg{n_j}}
%   = 2^{-\lambda -\tau + n_j}
%   \]

% \end{proof}

The following lemma relies on  \cite[Theorem~3.9 and Algorithm~5.3]{kirrinnis-joc-1998}.
\begin{lemma}[Modular representation]
  Let $F \in \CC[x]$ of degree $m$ and $\normi{F} \leq 2^{\tau_1}$.
  Let $P_j \in \CC[x]$ of degree $n_j$, for $ 1 \leq j \leq \nu$,
  such that $\sum_j n_j = n$ and $m \geq n$.
  Moreover, $2^{\rho}$ be an upper bound on the magnitude of the roots of all $P_j$.
  Let $\wt{F}$, resp. $\wt{P}_j$, be an approximation of $F$, resp. $P_j$,
  such that 
  $\normi{F - \wt{F}} \leq 2^{-\lambda}$,
  resp. $\normi{P_j - \wt{P}_j} \leq 2^{-\lambda}$.

  If $\lambda = \ell + \tau_1 + 2(n+m)\rho$ then we compute 
  approximations $\wt{F}_j$ of $F_j = F \mod P_j$,
  such that $\normi{F_j - \wt{F}_j} \leq 2^{-\ell}$
  in 
  $\OB( \mu( (m+n\lg{n})(\ell + \tau_1 + (m+n)\rho)) )$.  
\end{lemma}
\begin{proof}
  The polynomials $P_j(2^{\rho}\, x)$ have all their roots inside the unit disc.
%y
  %y which w
%y
We make them monic and then denote them $p_j$.
  It holds:
  \[
  \normi{P_j - \wt{P}_j} \leq 2^{-\lambda} \Rightarrow
  \normo{P_j(2^{\rho}\,x) - \wt{P}_j(2^{\rho}\,x)} \leq 2^{-\lambda + n_j\rho + \lg{n_j}} 
  \]
  \[
  \Rightarrow
  \normo{p_j - \wt{p}_j} \leq 2^{-\lambda + n_j\rho + \lg{n_j}+1}
  \enspace .
  \]

  We apply the same transformation to $F$.

  Let $f = 2^{-\tau_1 -m\rho - \lg{m}} \, F(2^{\rho}\, x)$, then 
  \[
  \normi{F} \leq 2^{\tau_1} \Rightarrow
  \normo{F} \leq 2^{\tau_1 + \lg{m}} \Rightarrow
  \normo{F(2^{\rho}\,x)} \leq 2^{\tau_1 + m\rho + \lg{m}}
  \]
  \[\Rightarrow
  2^{-\tau_1 -m \rho + \lg{m}}\normo{F(2^{\rho}\,x)} \leq 1 \Rightarrow
  \normo{f} \leq 1 
  \enspace .
  \]

  Now we can use \cite[Theorem~3.9]{kirrinnis-joc-1998}
  choosing $\lambda = \ell + \tau_1 + 2(n+m)\rho$
  to guarantee 
  $\normo{f \mod p_j - \wt{f}_j} = \normo{f_j - \wt{f}_j} \leq 2^{-\ell}$.
  The complexity of the procedure is 
  $\OB( \mu( (m+n\lg{n})(\ell + \tau_1 + (m+n)\rho)) )$.
\end{proof}

\begin{remark}
In the case where $m=n$
%!!! then 
the 
%yprevious 
latter bound
  becomes $\sOB(n(\ell + \tau_1 + n\rho))$.
\end{remark}

{\scriptsize
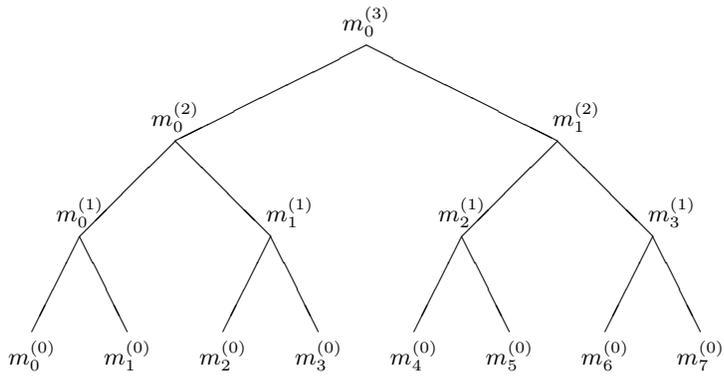
\begin{figure}[t]
\caption{Fan-in process, $m_j^{(h+1)}=m_{2j}^{(h)}m_{2j+1}^{(h)}$}
\label{fig:fan-in}
\begin{center}
\setlength{\unitlength}{0.5in}
\begin{picture}(8,4)
\put(4,3.5){\line(-2,-1){2}}
\put(4,3.5){\line(2,-1){2}}
\put(2,2.5){\line(-1,-1){1}}
\put(2,2.5){\line(1,-1){1}}
\put(6,2.5){\line(-1,-1){1}}
\put(6,2.5){\line(1,-1){1}}
\put(1,1.5){\line(-1,-2){0.5}}
\put(1,1.5){\line(1,-2){0.5}}
\put(3,1.5){\line(-1,-2){0.5}}
\put(3,1.5){\line(1,-2){0.5}}
\put(5,1.5){\line(-1,-2){0.5}}
\put(5,1.5){\line(1,-2){0.5}}
\put(7,1.5){\line(-1,-2){0.5}}
\put(7,1.5){\line(1,-2){0.5}}
\put(4,3.5){\makebox(0,0.5){$m_0^{(3)}$}}
\put(2,2.5){\makebox(0,0.5){$m_0^{(2)}$}}
\put(6.2,2.5){\makebox(0,0.5){$m_1^{(2)}$}}
\put(1,1.5){\makebox(0,0.5){$m_0^{(1)}$}}
\put(3.2,1.5){\makebox(0,0.5){$m_1^{(1)}$}}
\put(5,1.5){\makebox(0,0.5){$m_2^{(1)}$}}
\put(7.2,1.5){\makebox(0,0.5){$m_3^{(1)}$}}
\put(0,0){\makebox(1,0.5){$m_0^{(0)}$}}
\put(1,0){\makebox(1,0.5){$m_1^{(0)}$}}
\put(2,0){\makebox(1,0.5){$m_2^{(0)}$}}
\put(3,0){\makebox(1,0.5){$m_3^{(0)}$}}
\put(4,0){\makebox(1,0.5){$m_4^{(0)}$}}
\put(5,0){\makebox(1,0.5){$m_5^{(0)}$}}
\put(6,0){\makebox(1,0.5){$m_6^{(0)}$}}
\put(7,0){\makebox(1,0.5){$m_7^{(0)}$}}
\end{picture}
\end{center}
\end{figure}

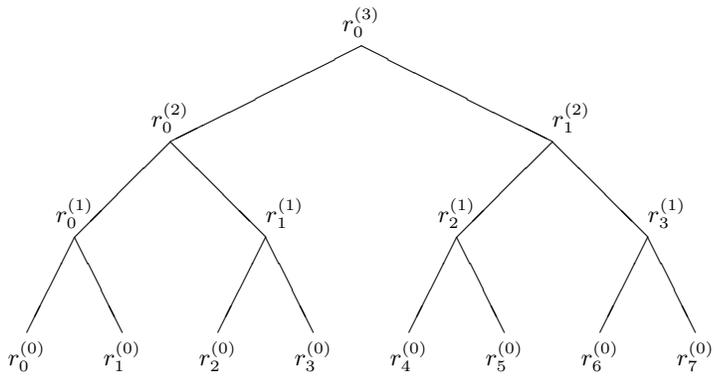
\begin{figure}[t]
\caption{Fan-out process, $r_j^{(h)}=r_{\lfloor j/2\rfloor}^{(h+1)}\bmod m_j^{(h)}=v(x)\bmod m_j^{(h)}$}
\label{fig:fan-out}
\begin{center}
\setlength{\unitlength}{0.5in}
\begin{picture}(8,4)
\put(4,3.5){\line(-2,-1){2}}
\put(4,3.5){\line(2,-1){2}}
\put(2,2.5){\line(-1,-1){1}}
\put(2,2.5){\line(1,-1){1}}
\put(6,2.5){\line(-1,-1){1}}
\put(6,2.5){\line(1,-1){1}}
\put(1,1.5){\line(-1,-2){0.5}}
\put(1,1.5){\line(1,-2){0.5}}
\put(3,1.5){\line(-1,-2){0.5}}
\put(3,1.5){\line(1,-2){0.5}}
\put(5,1.5){\line(-1,-2){0.5}}
\put(5,1.5){\line(1,-2){0.5}}
\put(7,1.5){\line(-1,-2){0.5}}
\put(7,1.5){\line(1,-2){0.5}}
\put(4,3.5){\makebox(0,0.5){$r_0^{(3)}$}}
\put(2,2.5){\makebox(0,0.5){$r_0^{(2)}$}}
\put(6.2,2.5){\makebox(0,0.5){$r_1^{(2)}$}}
\put(1,1.5){\makebox(0,0.5){$r_0^{(1)}$}}
\put(3.2,1.5){\makebox(0,0.5){$r_1^{(1)}$}}
\put(5,1.5){\makebox(0,0.5){$r_2^{(1)}$}}
\put(7.2,1.5){\makebox(0,0.5){$r_3^{(1)}$}}
\put(0,0){\makebox(1,0.5){$r_0^{(0)}$}}
\put(1,0){\makebox(1,0.5){$r_1^{(0)}$}}
\put(2,0){\makebox(1,0.5){$r_2^{(0)}$}}
\put(3,0){\makebox(1,0.5){$r_3^{(0)}$}}
\put(4,0){\makebox(1,0.5){$r_4^{(0)}$}}
\put(5,0){\makebox(1,0.5){$r_5^{(0)}$}}
\put(6,0){\makebox(1,0.5){$r_6^{(0)}$}}
\put(7,0){\makebox(1,0.5){$r_7^{(0)}$}}
\end{picture}
\end{center}
\end{figure}
}

\end{document}